\documentclass[twocolumn]{IEEEtran}

\usepackage{amsmath,graphics,amssymb,epsfig}
\usepackage{url}
\usepackage{color,soul}
\usepackage{verbatim}
\usepackage{amsthm}
\usepackage{tabularx}
\usepackage{cite}
\usepackage{multirow}

\usepackage{subcaption}
\newtheorem{remark}{Remark}
\newtheorem{prop}{Proposition}

\begin{document}
\setlength{\abovedisplayskip}{5pt}
\setlength{\belowdisplayskip}{5pt}

\title{{End-to-End Deep Learning for TDD MIMO Systems in the 6G Upper Midbands}}

\author{\IEEEauthorblockN{Juseong Park, Foad Sohrabi, Amitava Ghosh, and Jeffrey G. Andrews} \\ 
\thanks{This work has been supported by Nokia Bell Labs and in part by the National Science Foundation RINGS program grant CNS-2148141.}
\thanks{Juseong Park and Jeffrey G. Andrews are with 6G@UT in the Wireless Networking and Communications Group at the University of Texas at Austin, Austin, TX 78712, USA (email: juseong.park@utexas.edu, jandrews@ece.utexas.edu).}
\thanks{Foad Sohrabi with the Radio Systems Research at Nokia Bell Labs, New
Providence, NJ 07974 USA (e-mail: foad.sohrabi@nokia-bell-labs.com).}
\thanks{A. Ghosh is with the NOKIA Standards at Nokia Bell Labs, Naperville, IL 60563 (e-mail: amitava.ghosh@nokia-bell-labs.com).}
}
\maketitle


\begin{abstract}
This paper proposes and analyzes novel deep learning methods for downlink (DL) single-user multiple-input multiple-output (SU-MIMO) and multi-user MIMO (MU-MIMO) systems operating in time division duplex (TDD) mode. A motivating application is the 6G upper midbands (7-24 GHz), where the base station (BS) antenna arrays are large, user equipment (UE) array sizes are moderate, and theoretically optimal approaches are practically infeasible for several reasons.
To deal with uplink (UL) pilot overhead and low signal power issues, we introduce the channel-adaptive pilot, as part of the novel analog channel state information feedback mechanism. 
Deep neural network (DNN)-generated pilots are used to linearly transform the UL channel matrix into lower-dimensional latent vectors. 
Meanwhile, the BS employs a second DNN that processes the received UL pilots to directly generate near-optimal DL precoders.  The training is end-to-end which exploits synergies between the two DNNs.  For MU-MIMO precoding, we propose a DNN structure inspired by theoretically optimum linear precoding. The proposed methods are evaluated against genie-aided upper bounds and conventional approaches, using realistic upper midband datasets. Numerical results demonstrate the potential of our approach to achieve significantly increased sum-rate, particularly at moderate to high signal-to-noise ratio (SNR) and when UL pilot overhead is constrained. 
\end{abstract}


\section{Introduction} \label{sec:introduction}
Multiple-input multiple-output (MIMO) systems are in the midst of an interesting evolution due to the application of deep learning. Although the theory of MIMO systems is quite well understood at this point, optimum (or near-optimum) approaches often have important shortcomings in modern and future MIMO systems, with large antenna arrays. Examples include uplink (UL) overhead constraints which limit the amount of channel state information (CSI) feedback or UL pilots that can be sent, the use of analog or hybrid beamforming as opposed to fully digital beamforming, and the corresponding use of suboptimal precoders. Other hardware constraints and nonidealities such as quantization and nonlinear amplifiers erode the efficacy of supposedly optimal approaches. A well-designed deep learning approach can in principle adapt smoothly to all of these nonidealities.

The upper midbands, spanning carrier frequencies from roughly 7 to 24 GHz and referred to as Frequency Range 3 (FR3),
will likely be the most consequential new spectrum for 6G-era cellular systems \cite{Miao:23}. Focusing on the lower (e.g. 7-15 GHz) portion, this spectrum provides substantial new bandwidth along with favorable propagation compared to millimeter-wave (mmWave) bands \cite{Samsung:22}, \cite{FCCtac:23}, \cite{Kang:23}. Such carrier frequencies thus appear quite promising for increasing urban and suburban coverage and capacity to meet projected future demand \cite{ITU:23}. The antenna array dimensionality will increase substantially over conventional lower bands, while still allowing more digital beamforming and richer scattering as opposed to the mmWave bands. These novel characteristics make the upper midbands an interesting case study for the application of end-to-end (E2E) deep learning principles.  In this paper, we jointly design and learn user equipment (UE) and base station (BS) algorithms for UL pilot transmission, channel estimation, and downlink (DL) precoding, for a variety of single-user MIMO (SU-MIMO) and multi-user MIMO (MU-MIMO) scenarios. 

\subsection{Related Work}

Extensive research has focused on MIMO precoding, including recent applications of deep learning to MU-MIMO DL precoding to reduce computational complexity and approach optimal precoding in a wider range of scenarios.  
Theoretically, dirty paper coding  \cite{Caire:03} and UL-DL duality-based approaches \cite{Vishwanath:03} lead to optimal (nonlinear) precoding schemes. However, the iterative and idealized nature of these methods results in complex and ultimately suboptimal solutions when real-world constraints are considered.
Shifting our focus to more practical linear precoding methods, especially for MU-MIMO with single-antenna UEs, the optimal linear beamforming structure has been studied in \cite{Bjornson:14}.
Inspired by this structure, deep learning-based beamforming methods have been proposed in \cite{Xia:20} and \cite{Zhang:22}. However, extending these methods to scenarios with multiple antennas per UE is not straightforward, as that beamforming structure is no longer optimal. Some efforts in this direction include employing deep unfolding techniques to simplify the weighted minimum mean square error (WMMSE) algorithm \cite{Hu:21} and creating deep neural network (DNN) modules that directly compute WMMSE variables \cite{Jang:22}. 
In addition, a deep learning-based linear precoding method is proposed in \cite{Huang:20}, which leverages the UL-DL duality as discussed in \cite{Vishwanath:03, Shi:08}.

Turning our attention to MIMO channel estimation and pilot sequence design, least squares (LS) and linear minimum mean square error (LMMSE) estimation methods with given pilot sequences have been well-documented \cite{Biguesh:04, Sun:02, Binguesh:06}.  
Optimizing pilots over the specific channel distribution can enhance the accuracy of channel estimation, and this process can be efficiently facilitated using a data-driven approach through deep learning, as suggested in \cite{Chun:19}.

Precoding design can be further improved by considering the preceding channel estimation.
In time division duplex (TDD) mode, a DL precoding method employing LMMSE UL channel estimation is proposed in \cite{Jose:11}, based on channel reciprocity. While the design of precoders takes into account the probability distribution of channel estimation error, obtaining this error distribution accurately remains challenging in practical scenarios.
To address this, deep learning has been utilized to optimize the precoder in a data-driven manner. In \cite{Zhang:22}, an optimal structure-based MU-MIMO precoding method for single-antenna UEs is presented, considering the least square (LS) channel estimation. Similarly, in \cite{Lee:23}, a sequential optimization method consisting of adaptive-length random pilots, LS channel estimation, and precoding is proposed using DNNs.
Additionally, \cite{Attiah:22} proposes a deep learning-based hybrid beamforming method with LMMSE channel estimation.
However, the precoding performances in \cite{Zhang:22}, \cite{Lee:23}, and \cite{Attiah:22} are constrained due to the limitation of linear channel estimation approaches.

To fully leverage the potential of deep learning, the joint optimization of channel estimation and precoding has been the subject of recent studies \cite{Sohrabi:21, Jang:22}.
In \cite{Sohrabi:21}, an E2E deep learning approach for MU-MIMO systems with single-antenna UEs is presented. This approach includes DL channel estimation with optimized pilots, digital feedback, and DL precoding in FDD mode. Receiving DL pilots at UEs, the UE DNN encodes the pilots into bits, which are then perfectly received and used by the BS's DNN to compute the MU-MIMO precoder. 
This approach has been extended to systems with multiple-antenna UEs in \cite{Jang:22}, incorporating DNN modules inspired by the WMMSE method for precoding. However, digital feedback adopted in \cite{Sohrabi:21, Jang:22} necessitates that UEs accurately estimate the DL channels, albeit implicitly, and quantize the complex-valued channels into bits. 
This results in the need for the UE DNNs to have powerful computational capabilities, which is often limited in actual deployment.
Furthermore, since \cite{Sohrabi:21, Jang:22} assume perfect decoding of the feedback bits from users, any actual errors in CSI feedback might negatively affect precoding performance.

\subsection{Contributions}
This paper introduces E2E deep learning methodologies for UL pilot design, CSI acquisition, and DL precoding for SU-MIMO and MU-MIMO systems in TDD mode. 
Our four main contributions are outlined as follows.

\subsubsection{Efficient analog CSI feedback mechanism}
A novel concept, the \emph{channel-adaptive pilot}, is introduced for efficient analog CSI feedback.
Under the assumption that the UE has access to at least partial DL CSI, the UE's DNN adaptively generates pilots in response to the CSI. 
When transmitted over the UL, these pilots are multiplied with the MIMO channel matrix, thereby rendering the received pilots as a linearly compressed representation of the channel.
Regarding efficiency, traditional MIMO channel training methods usually require a number of pilots at least equal to the number of UE antennas for channel reconstruction. Our approach, in contrast, reduces the number of necessary pilots by compressing the CSI, preserving only the essential information for precoding.
Furthermore, unlike conventional pilots, there is no prerequisite for the BS to know the channel-adaptive pilots in advance. 
Lastly, our proposed channel-adaptive pilot demonstrates robustness to noise by projecting the channel onto the angular domain. This makes it well-suited for practical scenarios with low power UL signals, an important consideration.

\subsubsection{Joint CSI acquisition and precoding modules tailored for TDD systems} 
Upon receiving the noisy UL pilots, our proposed BS DNN performs joint CSI acquisition and DL precoding for either SU-MIMO or MU-MIMO systems. Since these pilots are directly designed for precoding, they may not provide enough information for full channel reconstruction. However, during the BS DNN's joint optimization process, only the essential parts of the CSI are implicitly recovered and utilized for DL precoding based on channel reciprocity. Accordingly, we refer to this process of extracting essential CSI for precoding as \emph{CSI acquisition}, rather than channel estimation.

\subsubsection{A theory-guided MU-MIMO precoding module} 
We propose a structured learning module for optimal MU-MIMO precoding, employing the identified necessary condition of optimal precoding as an inductive bias for training the BS DNNs. 
Specifically, based on the UL-DL duality principle, the optimal linear precoder is formulated in the form of LMMSE. 
As \cite{Jang:22} highlights, a naive DNN approach for MU-MIMO precoding can lead to significant sum-rate losses, particularly in high signal-to-noise ratio (SNR) ranges. However, by adopting this structured approach, we have observed significant improvements in MU-MIMO precoding performance in our experiments. We further simplify the computational complexity of this structure by reducing the dimensions involved in matrix inverse. 
Lastly, distinct from the mean square error (MSE) duality-based approach of \cite{Huang:20}, our method focuses on sum capacity duality and is tailored for joint CSI acquisition and precoding. 

\subsubsection{Application to 6G mid-band channels with two distinct datasets} 
To demonstrate the future relevance of our approaches, we applied them to upper mid-band channel data using two distinct datasets: a private 7 GHz band dataset provided by Nokia Bell Labs and a separate 7 GHz dataset generated by using the QuaDRiGa simulator \cite{Jae:21}. The Nokia ray-tracing dataset, which is compliant with the 3rd Generation Partnership Project (3GPP) specification TR 38.901 \cite{ETSI:20}, incorporates realistic aspects of wireless communications, including physical antennas, virtual antenna ports, polarizations, and other practical parameters. Experiments on these datasets establish the excellent performance of our method compared to both theoretical and practical benchmarks.


\section{System Model} \label{sec:System Model}
\subsection{DL Data Transmission and UL Pilots Models} \label{subsec:DL Data Transmission and UL Pilots Models}

Consider a DL multi-user MIMO system where a BS equipped with $N_t$ antennas serves $K$ UEs. It is assumed that the $k$-th UE is equipped with $N_{r,k}$ antennas and receives $N_{s,k}$ data streams from the BS. However, for simplicity, we assume $N_{r,k} = N_{r}$ and $N_{s,k} = N_{s}, \forall k$. We further assume that the system operates in TDD mode, and channel reciprocity holds. Letting the $k$-th UE's channel be ${\mathbf{H}}_k \in \mathbb{C}^{N_t \times N_r}$, the DL received signal at the $k$-th UE ${\mathbf{y}}^{\text{dl}}_k$ based on a narrowband signal model can be written as
\begin{align}
{{\mathbf{y}}^{\text{dl}}_k} = {\mathbf{H}}_k^H {\mathbf{x}}^{\text{dl}} + {\mathbf{n}}_k^{\text{dl}},
\label{eqn:dl_signalmodel}
\end{align}
where ${\mathbf{n}}_k^{\text{dl}} \sim \mathcal{CN}({\mathbf{0}}, \sigma^2_k {\mathbf{I}}_{N_r})$ represents the additive White Gaussian noise at UE $k$ with zero mean and covariance matrix $\sigma^2_k{\mathbf{I}}_{N_r}$, and ${\mathbf{x}}^{\text{dl}}$ indicates the transmitted DL signal.
Note that it is narrowband; thus, per subband (i.e., coherence band), frequency selective fading and correlation are ignored. Assuming that the BS adopts linear precoding, the transmitted signal ${\mathbf{x}}^{\text{dl}}$ can be expressed as
\begin{align}
{\mathbf{x}}^{\text{dl}} = \sum_{k=1}^{K} \mathbf{F}_k \mathbf{s}_k,
\label{eqn:dl_precodingmodel}
\end{align}
where $\mathbf{F}_k \in \mathbb{C}^{N_t \times N_s}$ denotes the precoding matrix for UE $k$ satisfying the total power constraint $\sum_{k=1}^{K} \text{Tr}\left({\mathbf{F}_k \mathbf{F}_k^H}\right) \leq \mathcal{E}_s$, and $\mathbf{s}_k \in \mathbb{C}^{N_s}$ indicates the transmitted data which satisfies $\mathbb{E}\left[\mathbf{s}_k \mathbf{s}_k^H\right] = \mathbf{I}_{N_s}$.

For CSI acquisition at the BS, analog CSI feedback is considered, where the BS computes the DL precoding matrices based on the received UL pilots. Assuming that the UL pilot signals are orthogonal in the time domain to each other across all UEs, the received UL pilot signal from UE $k$, denoted as ${\mathbf{Y}_k^{\text{ul}} \in \mathbb{C}^{N_t \times N_p}}$, is defined as 
\begin{align}
{{\mathbf{Y}}^{\text{ul}}_k} = {\mathbf{H}}_k {\mathbf{P}}_k + {\mathbf{N}}_k^{\text{ul}},
\label{eqn:ul_signalmodel}
\end{align}
where ${\mathbf{P}}_k \in \mathbb{C}^{N_r \times N_p}$ represents the pilot matrix of UE $k$ that satisfies the power constraint $\text{Tr}\left({\mathbf{P}_k \mathbf{P}_k^H}\right) \leq \mathcal{E}_p, \forall k$. Here, $N_p$ denotes the length of the pilots, which is in units of time domain symbols, e.g., the OFDM symbol period. Additionally, ${\mathbf{N}}_k^{\text{ul}} \in \mathbb{C}^{N_t \times N_p}$ represents a noise matrix for the UL pilot transmission. 

Typically, UE $k$ and the BS share knowledge of the pilot matrix $\mathbf{P}_k$ for channel estimation, since the channel estimation is a function of $\mathbf{P}_k$. In the proposed method, however, there is no need for the BS to possess knowledge of $ \mathbf{P}_k $. Furthermore, we explore the adaptive design of ${\mathbf{P}}_k$ based on ${\mathbf{H}}_k$, assuming it is known at the UE. This is captured by introducing a pilot matrix calculation function
\begin{align}
{\mathbf{P}}_k= f_{\text{UE}} \left({\mathbf{H}}_k \right).
\label{eqn:pilot_generation}
\end{align}
We will empirically demonstrate that the assumption of perfect knowledge of ${\mathbf{H}}_k$ at the UE can be relaxed to scenarios in which the UE possesses only partial information about ${\mathbf{H}}_k$, obtained through DL probing beams.

After receiving the UL pilots for all UEs, 
the BS computes the precoding matrices as
\begin{align}
\left[
\mathbf{F}_1, \cdots, \mathbf{F}_K
\right] = f_{\text{BS}} \left({{\mathbf{Y}}}^{\text{ul}}_1, \cdots, {{\mathbf{Y}}}^{\text{ul}}_K \right),
\label{eqn:DL_precodercomputation}
\end{align}
where $ f_{\text{BS}}(\cdot) $ denotes the precoding function, which also implicitly extracts CSI. Then, the achievable rate for UE $k$ in the context of the DL signal models \eqref{eqn:dl_signalmodel} and \eqref{eqn:dl_precodingmodel} is given by
\begin{align}
R_k = \log_2 \det \left( \mathbf{I}_{N_s} + \mathbf{F}_k^H \mathbf{H}_k \mathbf{C}_{\Tilde{{\mathbf{n}}}_k^{\text{dl}}}^{-1} \mathbf{H}_k^H \mathbf{F}_k \right),
\label{eqn:rate_model}
\end{align}
where $\mathbf{C}_{\Tilde{{\mathbf{n}}}_k^{\text{dl}}}$ denotes the effective DL noise-plus-interference covariance matrix for UE $k$, defined as
\begin{align}
\mathbf{C}_{\Tilde{{\mathbf{n}}}_k^{\text{dl}}} = \left( \sigma_k^2 {\mathbf{I}}_{N_r} + \sum_{i=1, i \ne k}^K \mathbf{H}_k^H \mathbf{F}_i \mathbf{F}_i^H \mathbf{H}_k \right).
\end{align}

\subsection{Problem Formulation}
\label{subsec:problemformulation}
As is typical for MU-MIMO, we attempt to maximize the achievable sum-rate for an analog feedback-based TDD system, 
yielding
\begin{subequations}
\begin{align}
\max_{f_{\text{UE}}(\cdot), f_{\text{BS}}(\cdot)}
\quad & \sum_{k=1}^{K}R_k,\\
\text{subject to} ~~ & \text{Tr}\left({\mathbf{P}_k \mathbf{P}_k^H}\right) \leq \mathcal{E}_p, \quad \forall k, \label{eqn:powerconst_P}\\
& \sum_{k=1}^{K}\text{Tr}\left({\mathbf{F}_k \mathbf{F}_k^H}\right) \leq \mathcal{E}_s. \label{eqn:powerconst_F}
\end{align}\label{eqn:problem}
\end{subequations}
This problem involves the joint design of the pilot matrix at the UE and the precoders at the BS with implicit channel estimation. The objective is to maximize the sum-rate while satisfying the pilot and precoder power constraints \eqref{eqn:powerconst_P} and \eqref{eqn:powerconst_F}.

Additionally, we desire to attain the same optimized sum-rate performance with the fewest possible number of pilots, $N_p$.
From the perspective of linear transformation, when $N_p < N_r$, multiplying ${\mathbf{H}}_k$ by ${\mathbf{P}}_k$ effectively reduces the column dimension of ${\mathbf{H}}_k$ from $N_r$ to $N_p$. 
This implies that the linear transformation represented by ${\mathbf{P}}_k$ does not permit a pseudo-inverse (a one-sided inversion) when applied to ${\mathbf{H}}_k$.
Consequently, with $N_p < N_r$, the least square channel estimator \cite{Kay:10} does not exist, and thus the UL analog feedback with ${\mathbf{P}}_k$ can be viewed as a lossy compression. Nonetheless, since our primary objective is not the precise reconstruction of ${\mathbf{H}}_k$ but the computation of effective precoders, we will see this is acceptable. Therefore, the genuine challenge lies in how to linearly compress the channel matrix, or equivalently, linearly transform the channel, by designing the UL pilot matrix, without sacrificing the information vital for computing the optimal precoder. In the following section, we show that there exists ${\mathbf{P}}_k$ without information loss for optimal SU-MIMO precoding if $N_p \geq N_s$. Moreover, we will highlight deep learning's ability to capture the essential information for optimal precoding, even with fewer pilots than $N_s$.

\section{Analog CSI Feedback for Joint CSI Acquisition and Precoding in SU-MIMO} \label{Analog CSI Feedback for Joint CSI Acquisition and Precoding in SU-MIMO}

In this section, we explore a deep learning-based approach for the joint design of UL pilots and DL precoding in the SU-MIMO scenario, within the context of analog CSI feedback.
While our focus is on SU-MIMO, Fig.~\ref{figure:proposedstructure} provides a generalized depiction of the proposed DNN structure for the MU-MIMO system. In this architecture, two DNN structures are employed---one at the BS and one at each UE, denoted as $f_{\text{BS}}(\cdot;\boldsymbol{\Theta}_{\text{BS}})$ and $f_{\text{UE}}(\cdot;\boldsymbol{\Theta}_{\text{UE}})$, respectively. Here, $\boldsymbol{\Theta}_{\text{BS}}$ and $\boldsymbol{\Theta}_{\text{UE}}$ represent the trainable parameters in the DNNs of the BS and UE, respectively. For the SU-MIMO case, where $K=1$, the user index subscript is omitted from the channel matrix, UL and DL signals, and the pilot and precoding matrices, among others. In the SU-MIMO setting, the absence of interference from other users simplifies the noise-plus-interference covariance matrix to $\mathbf{C}_{\Tilde{{\mathbf{n}}}^{\text{dl}}} = \sigma^2 {\mathbf{I}}_{N_r}$. As a result, the rate expressed in \eqref{eqn:rate_model} is now recast as
\begin{align}
R = \log_2 \det \left( \mathbf{I}_{N_s} + \frac{1}{\sigma^2} \mathbf{F}^H \mathbf{H} \mathbf{H}^H \mathbf{F} \right),
\label{eqn:singlerate}
\end{align}
Subsequently, the optimization problem in \eqref{eqn:problem} is adapted to maximize this revised $R$ while preserving the original power constraints on the pilot and the precoder.

\begin{figure}
\begin{center}
\includegraphics[width=0.9\linewidth]{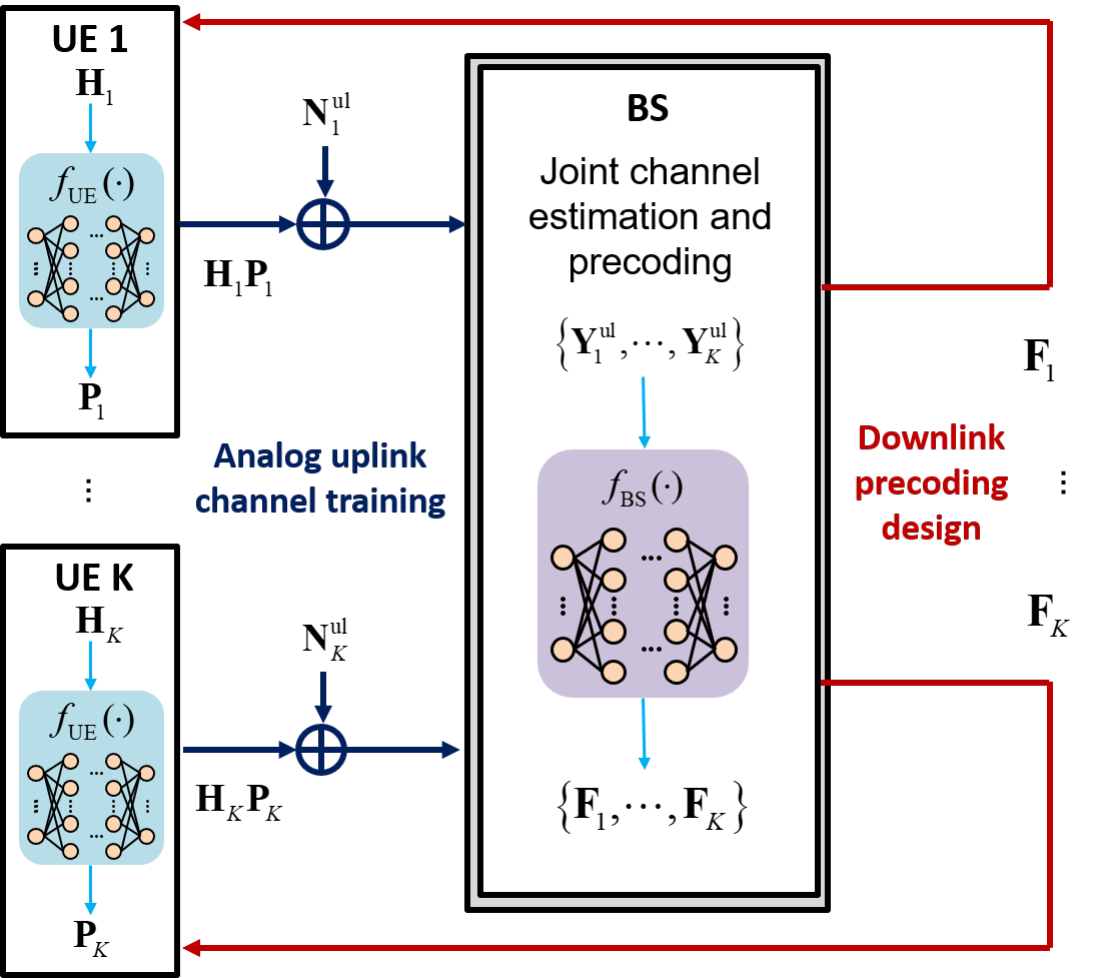}
\end{center}
\caption{The proposed deep neural network structure for UL analog channel training and DL precoding in MU-MIMO.}
\label{figure:proposedstructure}
\end{figure}


\subsection{Channel-Adaptive Pilot}

The deep learning-based channel-adaptive pilot for analog CSI feedback is detailed as follows. The UE DNN $f_{\text{UE}}(\cdot;\boldsymbol{\Theta}_{\text{UE}})$ generates the pilot matrix $\mathbf{P}$ as defined in \eqref{eqn:pilot_generation} assuming that DL channel $\mathbf{H}^H$ is given.
Thereafter, this matrix 
$\mathbf{P}$ is transmitted from the UE and received by the BS as in \eqref{eqn:ul_signalmodel}.

The proposed UE DNN $f_{\text{UE}}(\cdot;\boldsymbol{\Theta}_{\text{UE}})$ comprises two hidden fully connected (FC) layers followed by an output layer. The FC layers employ rectified linear unit (ReLU) activations, and each layer is equipped with batch normalization. The size of the output layer corresponds to twice the number of elements in the pilot matrix $\mathbf{P}$ for the real and imaginary parts of the matrix, and a linear activation function is used to account for potential negative values in $\mathbf{P}$. After the activation phase, the outputs are normalized to satisfy the power constraint \eqref{eqn:powerconst_P}.

Contrary to conventional systems, where the BS is informed about the specific pilot matrix $\mathbf{P}$ sent by the UE, our design eliminates this prerequisite; in other words, the UE does not explicitly convey to the BS which $\mathbf{P}$ is being used. This implies that the UE linearly compresses the channel matrix by multiplying it with $\mathbf{P}$, the specifics of how to compress, encapsulated in $\mathbf{P}$, are not disclosed to the BS.

While this process fundamentally involves linear transformation, the role of the UE DNN $f_{\text{UE}}\left(\cdot;\boldsymbol{\Theta}_{\text{UE}}\right)$ is crucial. This DNN computes the channel-adaptive pilot matrix $\mathbf{P}$, adapting to $\mathbf{H}^H$ through its inherent nonlinear approximation capabilities. In short, the channel-adaptive pilot is generated nonlinearly by the DNN and serves to compress the UL channel linearly.

In Section~\ref{sec:System Model}--\ref{subsec:problemformulation}, we noted that, from the linear transformation perspective, the condition $N_p \geq N_r$ is requisite for a complete reconstruction of the UE's UL channel matrix. However, when focusing on optimal precoding in the SU-MIMO system, the number of pilots $N_p$ can be less than the number of UE antennas $N_r$. 
Below is our proposition regarding the existence of the matrix $\mathbf{P}$, which preserves the core information for optimal precoding.

\begin{prop}\label{prop:1}
For a TDD-based noiseless UL analog feedback system, there exists a matrix $\mathbf{P}$ such that the optimal DL SU-MIMO precoder can be derived from the feedback $\mathbf{Y} = \mathbf{H} \mathbf{P}$, provided that $N_p \geq N_s$.
\end{prop}

\begin{proof}
When $\mathbf{P}$ is constructed using the right singular vectors of $\mathbf{H}$, there is no information loss from optimal precoding under a noiseless setup. For the full proof, see Appendix~\ref{app:proof_prop1}.
\end{proof}

Proposition 1 establishes, at least theoretically, that the UE DNN, coupled with the subsequent process of the BS DNN, has the potential to design a channel-adaptive pilot matrix for the system, especially when $N_p \geq N_s$. Intriguingly, these DNNs can further operate effectively even when $N_p < N_s$ by capturing the inherent characteristics of the channel. We present numerical evidence in Section~\ref{sec:Performance Analysis} to support this observation.

\subsection{Joint CSI Acquisition and Precoding}

In Fig.~\ref{figure:proposedstructure}, the BS DNN $f_{\text{BS}}\left(\cdot;\boldsymbol{\Theta}_{\text{BS}}\right)$ takes the noisy UL pilot signal ${{\mathbf{Y}}^{\text{ul}}}$, defined in \eqref{eqn:ul_signalmodel}, as its input. This network then directly produces precoders $\mathbf{F}$ crafted to minimize a loss function, to be elaborated upon in the following subsection.
The process can be written as
\begin{align}
    \mathbf{F} = f_{\text{BS}}\left(\mathbf{Y}^{\text{ul}};\boldsymbol{\Theta}_{\text{BS}}\right).
    \label{eqn: BS DNN for SU-MIMO}
\end{align} 

The architecture of BS DNN $f_{\text{BS}}(\cdot;\boldsymbol{\Theta}_{\text{BS}})$ is similar to that of the UE DNN $f_{\text{UE}}(\cdot;\boldsymbol{\Theta}_{\text{UE}})$, with two FC layers followed by an output layer. Both FC layers utilize ReLU activations and are equipped with batch normalization. A notable distinction lies in the size of the FC layers. Also, the output is normalized to satisfy \eqref{eqn:powerconst_F}.

Note that the DNN operates with noisy CSI as its input, thus the DNN acts as a nonlinear function approximator, tasked with mastering a 'many-to-one' mapping. To be specific, due to the noise, each instance of CSI can present in various forms, reflecting different noise conditions. Consequently, the DNN needs to learn how to map these diverse forms into a single, optimal precoder configuration.
However, as will be discussed in Section~\ref{sec:Performance Analysis}, our method effectively mitigates these challenges by projecting the channel matrix into the low dimensional angular domain with $\mathbf{P}$, illustrating its robustness to noise and thus highlighting the effectiveness of our proposed method in such challenging scenarios.

The learned precoder possibly differs from the well-known capacity-achieving precoder which is defined by the product of the right singular vectors of the DL channel matrix, $\mathbf{H}^H$, and the optimal power allocation matrix, $\mathbf{B}$, which can be mathematically represented as
\begin{align}
\mathbf{F}_{\text{opt}} = {\mathbf{U}}_{1:N_s} \mathbf{B}^{\frac{1}{2}}
\end{align}
where ${\mathbf{U}}$ is the right singular matrix of $\mathbf{H}^H$. Notably, multiplying any unitary matrix to the right side of $\mathbf{F}_{\text{opt}}$ will still result in a precoder that achieves the channel capacity. The key necessary condition for an optimal precoder is for its left singular matrix to match the right singular matrix of the channel $\mathbf{H}^H$. Given this, when directly optimizing \eqref{eqn:singlerate}, the learned precoder can be expressed as
\begin{align}
 \mathbf{F} = f_{\text{BS}}(\mathbf{Y}^{\text{ul}};\boldsymbol{\Theta}_{\text{BS}}) = {\mathbf{U}}_{1:N_s} \mathbf{B}^{\frac{1}{2}} \mathbf{C}
 \label{eqn:learned_SUMIMO_precoder}
\end{align}
where $\mathbf{C}$ is an arbitrary unitary matrix that the DNN adjusts based on each given $\mathbf{Y}^{\text{ul}}$.

\begin{remark}
As long as the left $N_s$ columns of $\mathbf{U}$ are chosen and water-filling (WF) power optimization is conducted over the corresponding singular values $\mu_1,\cdots,\mu_{N_s}$, the order of the columns of ${\mathbf{U}}_{1:N_s}$ does not impact the channel capacity. This implies that the order of the columns of ${\mathbf{U}}_{1:N_s}$ of the learned precoder can be varied across $\mathbf{H}$, similarly to $\mathbf{C}$.

\end{remark}


\subsection{Training Strategy}
Our proposed scheme is trained in an E2E manner to maximize the objective function defined in \eqref{eqn:problem}, utilizing the modified SU-MIMO rate expression in \eqref{eqn:singlerate}. To accomplish this, we construct a loss function $\mathcal{L}$ as
\begin{align}
\mathcal{L}({\boldsymbol{\Theta}}_{\text{UE}}, \boldsymbol{\Theta}_{\text{BS}}) \!=\! - {\mathbb{E}}_{\mathbf{H}} \left[ \log_2 \det \left( \mathbf{I}_{N_s} + \frac{1}{\sigma^2} \mathbf{F}^H \mathbf{H} \mathbf{H}^H \mathbf{F} \right) \right]
\end{align}
where we take the expectation over the distribution of $\mathbf{H}$, to enable our model to adapt effectively across diverse channel realizations. 

To update the trainable variables ${\boldsymbol{\Theta}}_{\text{UE}}$ and $\boldsymbol{\Theta}_{\text{BS}}$, we use stochastic gradient descent (SGD), wherein the sample mean is computed over the mini-batch. The training process can be expressed as
\begin{align} 
&\boldsymbol{\Theta }_{i}\leftarrow \boldsymbol{\Theta }_{i} \nonumber\\
&-\,\xi \nabla _{\boldsymbol{\Theta }_{i}} \!\!\left({\frac{1}{M_B}}\sum_{j=1}^{M_B} \left[ \log_2 \det \left( \mathbf{I}_{N_s} + \frac{1}{\sigma^2} \mathbf{F}_j^H \mathbf{H}_j \mathbf{H}_j^H \mathbf{F}_j \right) \right] \right),\label{eqn:sgd}
\end{align}
where $i$ represents either the BS or the UE. Here, $\xi$ denotes the learning rate, and $M_B$ is the number of elements in a mini-batch, respectively.

\section{MU-MIMO Systems} \label{sec:MU-MIMO Systems}
We now extend the proposed scheme to MU-MIMO systems. The key difference between the sum-rate maximization problem, as outlined in \eqref{eqn:problem}, and the previous SU-MIMO achievable rate maximization problem, lies in the importance of inter-user interference (IUI) terms. Also, we consider a generalized IUI setting where each user has multiple streams. The channel feedback requirements for optimal precoding in MU-MIMO might be expected to differ from the SU-MIMO case, on account of the role of IUI.

However, the structure of the UE DNN for the UL channel-adaptive pilot in the MU-MIMO system is similar to its counterpart in the SU-MIMO case. We assume that each UE utilizes an identical DNN architecture and also shares the same weights, which greatly simplifies training. 
Note that using the same weights does not produce universal pilots across UEs, as the pilots are a function of the CSI. Therefore, the DNN for UE $k$ can be expressed as
\begin{align}
\mathbf{P}_k = f_{\text{UE}}(\mathbf{H}_k;\boldsymbol{\Theta}_{\text{UE}}).
\label{eqn:UE DNN in SU-MIMO}
\end{align}
In this equation, $f_{\text{UE}}(\cdot)$ and $\boldsymbol{\Theta}_{\text{UE}}$ are not subscripted with $k$, highlighting that all UEs share the same pre-trained DNN, which can be thought of as representative of UEs over the entire cell.

\subsection{Naive DNN-Based Joint CSI Acquisition and Precoding Module}
The first BS DNN module we consider is a direct extension of the one used for SU-MIMO, as represented in \eqref{eqn: BS DNN for SU-MIMO}. Upon receiving the UL pilots from all $K$ users, the BS computes the linear precoders for each of the $K$ UEs. This process can be articulated as
\begin{align}
\left[ \mathbf{F}_1, \cdots, \mathbf{F}_K \right] = f_{\text{BS}} \left(\mathbf{Y}_1^{\text{ul}}, \cdots, \mathbf{Y}_K^{\text{ul}};\boldsymbol{\Theta}_{\text{BS}} \right).
\label{eqn:DL_precodercomputation_dnn}
\end{align}

Unfortunately, as highlighted in \cite{Jang:22}, the naive DNN, which is a full DNN approach directly computing precoders, exhibits a significant sum-rate loss as the SNR increases. This phenomenon arises because the naive DNN might be proficient in power allocation but struggles to effectively mitigate IUI. Note that while power allocation between users predominantly determines the sum-rate performance at low SNR, suppressing IUI becomes increasingly important at higher SNR. 

\subsection{Optimal Structure-Based Joint CSI Acquisition and Precoding Module}

\begin{figure*}[t] 
    \centering
    \includegraphics[width=0.8\textwidth]{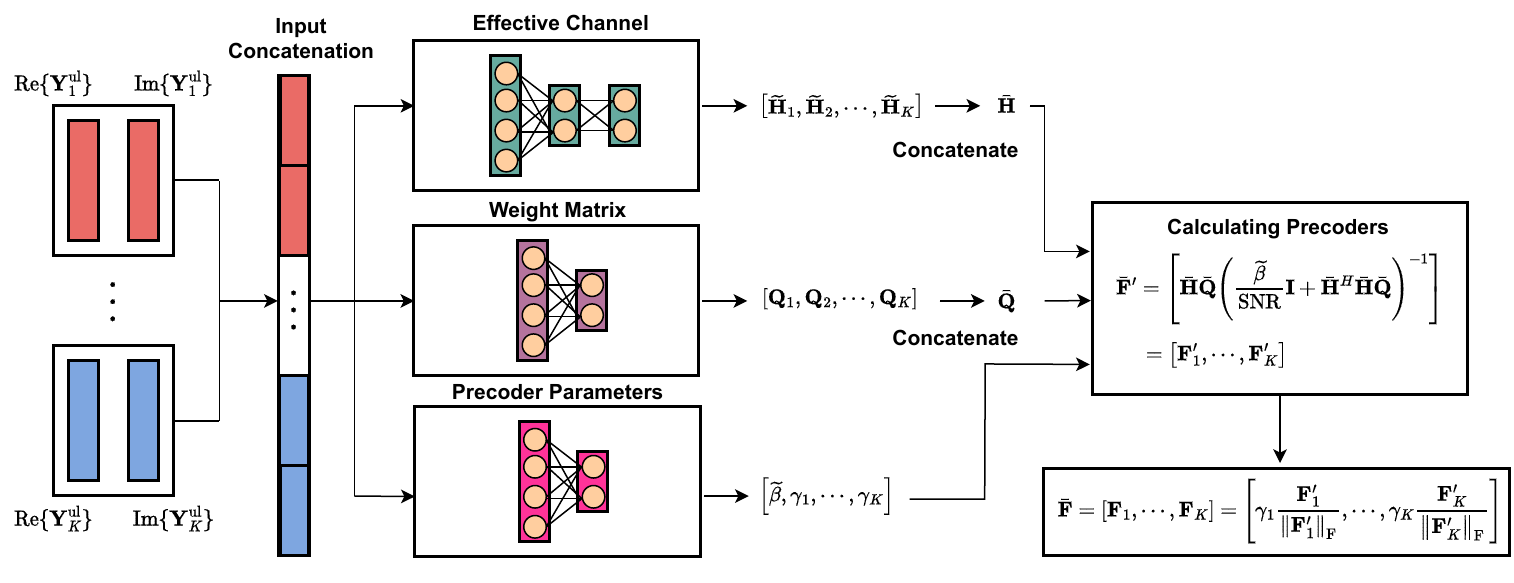}
    \caption{The proposed MU-MIMO scheme for joint CSI acquisition and precoding.}
    \label{figure:proposed_mumimo_structure}
\end{figure*}

To address the performance degradation observed in naive DNNs, we investigate the optimal structure of the MU-MIMO DL linear precoder that maximizes sum-rate, aiming to incorporate this structure into DNNs as an inductive bias. First, we simplify \eqref{eqn:problem}, which comprises multiple subproblems.

To focus on the MU-MIMO DL precoding part, we introduce two assumptions to reduce it to the standard MU-MIMO DL precoding problem. First, we assume that the function $f_\text{UE}(\cdot)$ is given. Second, assuming the channel matrices have been estimated, we directly compute the precoder $\mathbf{F}_k$.
Under these assumptions, the original problem in \eqref{eqn:problem} simplifies to
\begin{subequations}
\begin{align}
\max_{\mathbf{F}_1, \cdots, \mathbf{F}_K}
\quad & \sum_{k=1}^{K}R_k,\\
\text{subject to} ~~ & \sum_{k=1}^{K}\text{Tr}({\mathbf{F}_k \mathbf{F}_k^H}) \leq \mathcal{E}_s\label{eqn:powerconst_F_2}
\end{align}\label{eqn:MU-MIMO_BC_problem}
\end{subequations}
where $R_k$ is as defined in \eqref{eqn:rate_model}.

In the MU-MIMO DL sum-rate maximization problem outlined by \eqref{eqn:MU-MIMO_BC_problem}, the non-convex nature of the problem complicates the identification of the optimal solution structure. However, by utilizing the duality principles \cite{Vishwanath:03, Zhang:12}, the structure can be revealed through its corresponding dual UL problem. The dual UL problem for a given MU-MIMO DL is defined by the conjugate transpose of the DL's channel, with the power constraint for the DL transmitter being equal to the combined power constraints of all UL transmitters. Importantly, both the DL and its dual UL are shown to share the same capacity region \cite{Vishwanath:03}. This relationship persists even under linear processing conditions \cite{Shi:08}. Therefore, considering that the LMMSE receiver is the optimal linear choice in the dual UL problem associated with \eqref{eqn:MU-MIMO_BC_problem}, it follows that LMMSE also represents the optimal linear precoding for MU-MIMO DL systems.

However, this does not imply that the precoder of the DL is identical to the decoder of the dual UL. The distinction arises because, according to Chap. 9.5 of \cite{Heath:18}, power allocations in the two problems may differ, but this can be addressed by introducing linear filters. Building upon this duality, Proposition 2 introduces a necessary condition for optimal MU-MIMO DL linear precoding.

\begin{prop}\label{prop:2}
The optimal DL linear precoder for MU-MIMO is characterized by the LMMSE with reweighting.
In other words, in the MU-MIMO DL context, by denoting $k$-th UE's receiver and weight matrix as $\mathbf{W}_k$ and $\mathbf{Q}_k$, we have the optimal linear precoder as
\begin{align}
\mathbf{F}_k &= \gamma \! \left[ \sum_{m=1}^{K} \!\left( \beta_m\mathbf{I}_{N_t} \!\!+\! \mathbf{H}_m \mathbf{W}_m^H \mathbf{Q}_m \mathbf{W}_m \mathbf{H}_m^H\right) \right]^{-1} \!\!\!\!\! \mathbf{H}_k \mathbf{W}_k^H \mathbf{Q}_k \nonumber
\\
&= \gamma \! \left[ \beta\mathbf{I}_{N_t} + \sum_{m=1}^{K} \mathbf{H}_m \mathbf{W}_m^H \mathbf{Q}_m \mathbf{W}_m \mathbf{H}_m^H \right]^{-1} \!\!\!\!\! \mathbf{H}_k \mathbf{W}_k^H \mathbf{Q}_k  \label{eqn:MU-MIMO_optimal_structure}
\end{align}
where $\beta_m = \frac{\sigma_m^2}{\mathcal{E}_s} \mathrm{Tr}\left(\mathbf{W}_m^H \mathbf{Q}_m \mathbf{W}_m\right)$, $\beta = \sum_{m=1}^K {\beta_m}$ and $\gamma$ is a normalization factor ensuring $\sum_{k=1}^{K}\mathrm{Tr}\left({\mathbf{F}_k \mathbf{F}_k^H}\right) \leq \mathcal{E}_s$. 
\end{prop}

\begin{remark}
The optimal structure defined by \eqref{eqn:MU-MIMO_optimal_structure} has been partially adopted in many existing algorithms, such as the WMMSE algorithm \cite{Christensen:08, Shi:11}. Specifically, the stationary condition of the weighted sum mean square error minimization problem is equivalent to \eqref{eqn:MU-MIMO_optimal_structure} (See (15) of \cite{Shi:11}). As a result, several deep learning methods implementing the WMMSE algorithm, such as those in \cite{Jang:22, Hu:21}, achieve favorable MU-MIMO DL sum-rate performance due to their adoption of this structure.
\end{remark}

We further simplify \eqref{eqn:MU-MIMO_optimal_structure} by defining the effective channel of UE $k$ as $\widetilde{\mathbf{H}}_k = \mathbf{H}_k \mathbf{W}_k^H$, to obtain
\begin{align}
\mathbf{F}_k= \gamma \! \left[ \beta\mathbf{I}_{N_t} + \sum_{m=1}^{K} \widetilde{\mathbf{H}}_m \mathbf{Q}_m \widetilde{\mathbf{H}}_m^H \right]^{-1} \widetilde{\mathbf{H}}_k^H \mathbf{Q}_k. \label{eqn:MU-MIMO_optimal_structure_simplified}
\end{align}
Note that \eqref{eqn:MU-MIMO_optimal_structure_simplified} entails a matrix inverse, for which the computational complexity scales as $\mathcal{O}(N_t^3)$. In massive MIMO configurations, such complexity can be prohibitive. To address this challenge, we leverage the matrix inversion lemma \cite{Petersen:08}, and introduce the following stacked matrices: 
$\bar{\mathbf{F}} = \left[\mathbf{F}_1, \dots ,\mathbf{F}_K\right] \!\in\! \mathbb{C}^{N_t\times KN_s}$,  $\bar{\mathbf{H}} = \left[\widetilde{\mathbf{H}}_1 \dots \widetilde{\mathbf{H}}_K\right] \!\in\!\mathbb{C}^{N_t\times KN_s}$, and $\bar{\mathbf{Q}} = \mathrm{blkdiag}\left(\mathbf{Q}_1, \dots , \mathbf{Q}_K\right)\in \mathbb{C}^{KN_s\times KN_s}$ where $\mathrm{blkdiag}\left(\cdot\right)$ represents the block diagonal matrix.
Using these, \eqref{eqn:MU-MIMO_optimal_structure_simplified} can be rewritten as 
\begin{align}
\bar{\mathbf{F}} = \gamma \! \left[ \beta\mathbf{I}_{N_t} + \bar{\mathbf{H}}\bar{\mathbf{Q}}\bar{\mathbf{H}}^H \right]^{-1} \bar{\mathbf{H}}\bar{\mathbf{Q}}.
\end{align}
Employing the matrix inversion lemma, we obtain
\begin{align}
\bar{\mathbf{F}} = \gamma \bar{\mathbf{H}}\bar{\mathbf{Q}}\left[ \beta\mathbf{I}_{KN_s} + \bar{\mathbf{H}}^H \bar{\mathbf{H}}\bar{\mathbf{Q}}\right]^{-1}.
\label{eqn:MU-MIMO_optimal_structure_manipulated}
\end{align}
Importantly, the computational complexity of this matrix inversion is of the order $\mathcal{O}((KN_s)^3)$, where $KN_s$ is less than or equal to $N_t$. In massive MIMO scenarios, $N_t/N_r \gg K$ and then $N_t \gg KN_s$, so the difference is significant.

Now, we propose a BS DNN incorporating the optimal structure for DL MU-MIMO precoding, as in \eqref{eqn:MU-MIMO_optimal_structure_manipulated}. The proposed MU-MIMO scheme is presented in Fig.~\ref{figure:proposed_mumimo_structure}. The network receives the UL signal matrices ${\mathbf{Y}_1^{\text{ul}}, \cdots, \mathbf{Y}_K^{\text{ul}}}$ as inputs and concatenates them into a single vector.

The subsequent BS DNN is composed of three specialized neural networks. The \textit{Effective Channel DNN} computes the set of effective channels ${\widetilde{\mathbf{H}}_1, \dots, \widetilde{\mathbf{H}}_K}$. Concurrently, the \textit{Weight Matrix DNN} outputs the corresponding weight matrices ${\mathbf{Q}_1,\cdots,\mathbf{Q}_K}$, and the \textit{Precoder Parameters DNN} determines the necessary precoder parameters ${\widetilde{\beta}, \gamma_1, \cdots, \gamma_K}$. 

These parameters are utilized to construct the preliminary precoding matrix $\bar{\mathbf{F}}'$ according to the optimal structure \eqref{eqn:MU-MIMO_optimal_structure_manipulated} as
\begin{align}
\bar{\mathbf{F}}' &= \left[\bar{\mathbf{H}}\bar{\mathbf{Q}} \left( \frac{\widetilde{\beta}}{\text{SNR}}\mathbf{I} +\bar{\mathbf{H}}^H\bar{\mathbf{H}}\bar{\mathbf{Q}} \right)^{-1} \right] = \left[ \mathbf{F}'_1, \cdots, \mathbf{F}'_K \right].
\end{align}
To explicitly optimize power allocation between UEs, we employ individual UE power allocation factors ${ \gamma_1, \cdots, \gamma_K }$ satisfying $\sum_{k=1}^K \gamma_k^2 = \mathcal{E}_s$. This results in the final precoder matrix $\bar{\mathbf{F}}$, where each UE's precoder is scaled individually as
\begin{align}
\bar{\mathbf{F}} =\left[ \mathbf{F}_1, \cdots, \mathbf{F}_K \right]=\left[ \gamma_1\frac{\mathbf{F}'_1}{\left\|\mathbf{F}'_1 \right\|_\mathrm{F}}, \cdots, \gamma_K \frac{\mathbf{F}'_K}{\left\|\mathbf{F}'_K \right\|_\mathrm{F}}\right].
\end{align}
Based on the optimal structure, this tailored approach for precoding ensures an optimized allocation of power across the network, adequately suppressing IUI.

However, we cannot guarantee that the output of the \textit{Effective Channel DNN} will converge to $\widetilde{\mathbf{H}}_k = \mathbf{H}_k \mathbf{W}_k^H$, and this limitation also applies to the other DNNs. Given our system's E2E training with the sum-rate loss and the unknown optimal values for each element, training these DNNs to produce outputs close to the optimal is challenging. Nonetheless, this structure provides an inductive bias, a concept validated in machine learning literature, such as in the dueling architecture of deep Q-networks \cite{Wang:16}.

In terms of the DNN architecture, we employed two fully connected (FC) layers for both the \textit{Weight Matrix DNN} and the \textit{Precoder Parameters DNN}, while the \textit{Effective Channel DNN} incorporates three FC layers due to its larger output dimension. Regarding the activation functions for the output layers, linear activation functions are used, except for the \textit{Precoder Parameters DNN}. For the \textit{Precoder Parameters DNN}, we apply linear activation for $\beta$ and softmax activation for ${\gamma_1, \cdots, \gamma_K}$, as $\gamma$ represents the allocated power ratio. Our design philosophy prioritizes simplicity in the communication system architecture, hence the choice not to use more complex structures such as CNNs, residual connections \cite{He:16}, Vision Transformers \cite{Dosovitskiy:20}, or other advanced deep learning structures \cite{Arvinte:23} — these could be considered in future work.

Training the E2E system for MU-MIMO closely mirrors the approach used for SU-MIMO. Thus, the entire system is trained E2E to maximize the sum-rate, utilizing SGD as outlined in \eqref{eqn:sgd}. However, it is important to note that the loss function in this context is focused on the sum-rate, rather than on channel capacity.

\section{Performance Analysis} \label{sec:Performance Analysis}
Having now defined and motivated our approach, we assess the performance of the proposed deep learning-enhanced feedback method, as well as the joint CSI acquisition and precoding scheme, through numerical examples for both TDD SU-MIMO and MU-MIMO scenarios. 
Throughout this section, we assume the DL noise variances at UEs to be identical, i.e., $\sigma_k^2 = \sigma^2, \forall k$, for simplicity. Similarly, letting $\tilde{\sigma}_k^2$ be the UL noise variance for UE $k$, we assume $\tilde{\sigma}_k^2 = \tilde{\sigma}^2, \forall k$. Then, the DL SNR and UL SNR are defined as ${\mathcal{E}_s}/{\sigma^2}$ and ${\mathcal{E}_p}/{\tilde{\sigma}^2}$, respectively.

\begin{table}[t]
\centering
\caption{Dataset Specifications}
\label{table:specifications}
\begin{tabular}{|c|c|}
\hline
\textbf{Parameter} \rule{0pt}{2.0ex} & \textbf{Value} \\ 
\hline
\hline
Cell type & Single cell \\
\hline
Cell radius & 100 m \\
\hline
BS position & (0, 0, 25) m \\
\hline
Channel model & UMa in TR 38.901 \cite{ETSI:20} \\
\hline
Carrier frequency & 7 GHz \\
\hline
Simulation Bandwidth & 20 MHz \\
\hline
Subcarrier spacing & 30 kHz \\
\hline
Resource blocks (RBs) & 52 \\
\hline
Subcarriers per RB & 12 \\
\hline
\end{tabular}
\end{table}

\subsection{Dataset Specifications}
We utilized two distinct datasets for our analysis: Nokia's private mid-band dataset and a dataset created by the QuaDRiGa simulator \cite{Jae:21}. The key specifications for both of these datasets are detailed in Table \ref{table:specifications}. We assumed that the resource block (RB) scheduling is predetermined, and as a result, random RB's channel data is leveraged for both training and testing. 
For the antenna setups, the Nokia BS array is equipped with 32 virtual ports, where the number 32 corresponds to $2 \times 8 \times 2 = \text{rows} \times \text{columns} \times \text{polarizations}$. Similarly, the Nokia UE array has 4 virtual ports, calculated as $1 \times 2 \times 2$. For the QuaDRiGa's setup, instead of considering virtual ports, we directly focused on physical antennas. Specifically, the QuaDRiGa BS is equipped with 32 antennas, determined by the configuration $4 \times 8 \times 1$, and the UE features 4 antennas, configured as $1 \times 4 \times 1$. Since a narrowband is assumed, each generated $\mathbf{H}_k \in\mathbb{C}^{32\times4}$ corresponds to a channel matrix for a specific RB.

For the creation of our training and testing samples, we started with 20,000 initial UE realizations. Out of these, we randomly selected 18,000 UEs for the training set, with the remaining 2,000 UEs forming the test set. In the SU-MIMO scenario, from a total of 936,000 channel samples (calculated as $18,000~\text{UEs} \times 52~ \text{RBs}$), we randomly chose 50,000 for training and 2,000 for testing.
For MU-MIMO samples, we created each one by picking four UE indices and one RB index at random. This approach allowed us to combine four channel samples into a single training or test sample. We kept the number of training samples the same as in the SU-MIMO case, at 50,000, and the test samples at 2,000.

\subsection{Baseline Methods}
\subsubsection{UL channel estimation}
We focus on two baseline methods for UL channel estimation--the linear minimum mean square error (LMMSE) estimation and the regularized least square (RLS) estimation. For the both LMMSE and RLS channel estimation methods, the BS should know what $\mathbf{P}$ was sent from the UE whereas the BS does not need to know $\mathbf{P}$ in our proposed method.

\textbf{LMMSE method}: Starting with the UL pilot signal representation $\mathbf{Y}=\mathbf{H}\mathbf{P}+\mathbf{N}$,
we first vectorize it using a similar technique described in Section V of \cite{Heath:16} as
\begin{align}
\mathrm{vec}(\mathbf{Y})=({\mathbf{P}^T}\otimes\mathbf{I}_{N_t})\mathrm{vec}(\mathbf{H}) +\mathrm{vec}(\mathbf{N}), 
\label{eqn:vectorized_UL}
\end{align}
where $\otimes$ indicates the Kronecker product operation.
To simplify notation, define $\tilde{\mathbf{y}}$ = $\mathrm{vec}(\mathbf{Y}) \in \mathbb{C}^{N_tN_p}$, $\tilde{\mathbf{P}} = {\mathbf{P}^T}\otimes\mathbf{I}_{N_t}\in \mathbb{C}^{N_tN_p \times N_tN_r}$, $\tilde{\mathbf{h}} =\mathrm{vec}(\mathbf{H}) \in \mathbb{C}^{N_tN_r}$, and $\tilde{\mathbf{n}} = \mathrm{vec}(\mathbf{N})\in \mathbb{C}^{N_tN_p}$. Utilizing these definitions, \eqref{eqn:vectorized_UL} simplifies to $
\tilde{\mathbf{y}}=\tilde{\mathbf{P}}\tilde{\mathbf{h}}+\tilde{\mathbf{n}}$.
Then, according to \cite{Kay:10}, the LMMSE estimate for $\tilde{\mathbf{h}}$ can be expressed as
\begin{align}
\tilde{\mathbf{h}}_\mathrm{MMSE}=\mathbf{A}^H_\mathrm{MMSE}(\tilde{\mathbf{y}}-\mathbb{E}[\tilde{\mathbf{y}}]) + \mathbb{E}[\tilde{\mathbf{h}}]
\end{align}
where the matrix $\mathbf{A}^H_\mathrm{MMSE}$ is given by
\begin{align}
\mathbf{A}^H_\mathrm{MMSE} = {\mathbf{C}_{\tilde{\mathbf{h}}}}\tilde{\mathbf{P}}^H(\tilde{\mathbf{P}}{\mathbf{C}_{\tilde{\mathbf{h}}}}\tilde{\mathbf{P}}^H + \mathbf{C}_{\tilde{\mathbf{n}}})^{-1}.
\label{eqn:lmmse_channelestimation}
\end{align}
Here, ${\mathbf{C}_{\tilde{\mathbf{h}}}} = \mathbb{E}[\tilde{\mathbf{h}}\tilde{\mathbf{h}}^H]$ and ${\mathbf{C}_{\tilde{\mathbf{n}}}} = \mathbb{E}[\tilde{\mathbf{n}}\tilde{\mathbf{n}}^H]=\tilde{\sigma}^2\mathbf{I}_{N_tN_p}$.
Finally, we can obtain the estimate of the UL channel $\mathbf{H}$ by devectorizing $\tilde{\mathbf{h}}_\mathrm{MMSE}$.

Note that for the implementation of the LMMSE channel estimation, the BS must be aware of the channel statistics, specifically, ${\mathbf{C}_{\tilde{\mathbf{h}}}}$ and $\mathbb{E}[\tilde{\mathbf{h}}]$, which are challenging to obtain in practice. Thus, we consider the LMMSE method as a strong baseline that requires more statistical information and even the knowledge of the pilot matrix at the BS, in contrast to our proposed method.

\textbf{RLS method}: In contrast to the LMMSE estimation, the RLS method does not require knowledge of channel statistics. For the UL pilot signal $\mathbf{Y}=\mathbf{H}\mathbf{P}+\mathbf{N}$, the RLS channel estimate $\hat{{\mathbf{H}}}_{\text{RLS}}$ is obtained as
\begin{align}
\hat{{\mathbf{H}}}_{\text{RLS}} = \mathbf{Y}{\mathbf{P}^H} \left( \mathbf{P}{\mathbf{P}^H} + \delta\mathbf{I}_{N_p}\right)^{-1},
\label{eqn:RLS}
\end{align}
where $\delta$ is set to $\tilde{\sigma}^2$. Since \eqref{eqn:RLS} includes the regularization term, there always exists the inverse regardless of the rank of $\mathbf{P}$. In our simulations, the RLS method is considered as a practical baseline.

\subsubsection{UL pilot signals}
We can consider two different types of pilots as baseline methods: the Walsh code-based pilot and the SVD-based pilot.

\textbf{Walsh code-base pilot}:
The Walsh code-based pilot, characterized by its orthogonality, employs a real-valued length 4 Walsh code. 

\textbf{SVD-based pilot}:
The SVD-based pilots are derived from Proposition~\ref{prop:1} and utilize the $N_p$ right singular vectors of the UL channel matrix $\mathbf{H}$ in descending order. As with the use of SVD for low-rank approximation, employing the right singular vectors as pilots captures the significant components of the given channel matrix, ensuring the orthogonality of the pilots. 

\subsubsection{DL precoding schemes}
Based on the estimated channel matrix, the baseline method employs \textbf{SVD-based beamforming (BF)} combined with WF power allocation, which is known to achieve capacity in the SU-MIMO scenario. To be specific, the right singular vectors of the estimated DL channel matrix $\mathbf{H}^H$ are used for precoding with WF power allocation. 

For MU-MIMO, two precoding methods are considered as benchmarks: \textbf{WMMSE BF} with matched filter initialization and 20 iterations \cite{Christensen:08}, and \textbf{block diagonalization (BD) BF} with per-user WF power allocation \cite{Spencer:04}. The WMMSE method, well-known for achieving near-optimal achievable sum-rate performance in linear processing, is iterative and requires extensive matrix computations. Therefore, the BD method is employed as an additional benchmark, offering a more practical approach.

\subsection{Performance Comparison to Baseline Methods}

\begin{figure}
\centering
\begin{subfigure}[t]{\linewidth}
    \centering
    \includegraphics[width=\linewidth]{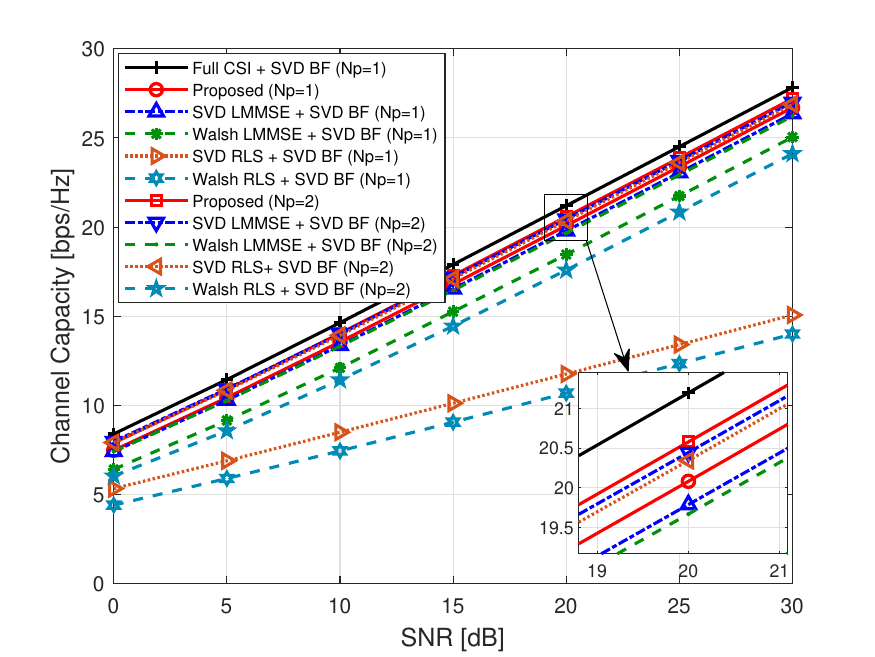}
    \caption{Using Nokia's mid-band dataset.}
    \label{figure:DL_channelcapacity_nokia}
\end{subfigure}
\newline  
\begin{subfigure}[t]{\linewidth}
    \centering
    \includegraphics[width=\linewidth]{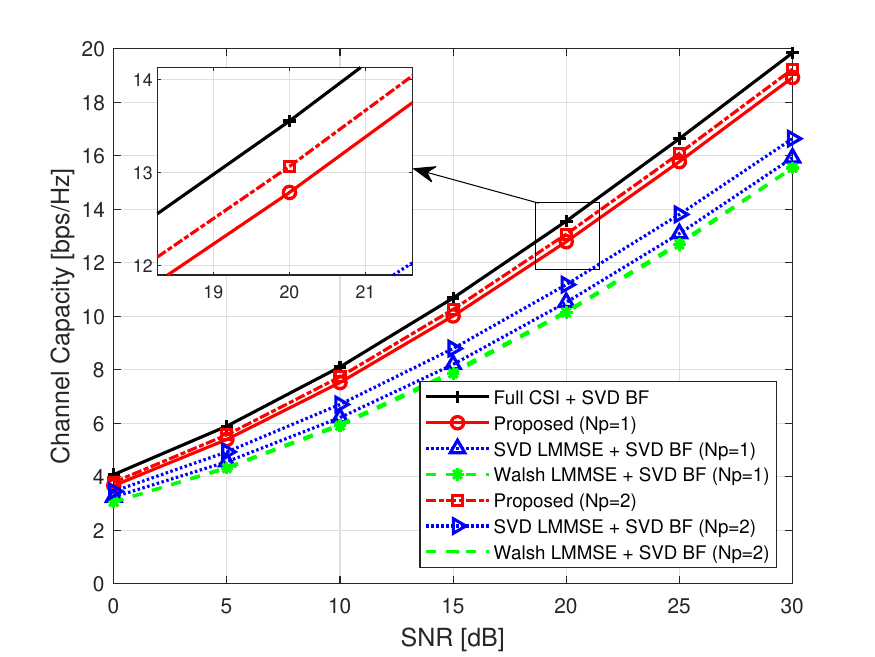}
    \caption{Using the QuaDRiGa mid-band dataset.}
    \label{figure:DL_channelcapacity_quadriga}
\end{subfigure}
\caption{DL SU-MIMO channel capacity performance versus DL SNR for $ N_p=\{1, 2\} $ and $ N_s=2 $ with fixed UL $ \text{SNR}=10\, \text{dB} $.}
\label{figure:DL_channelcapacity_comparison}
\end{figure}

\subsubsection{SU-MIMO DL channel capacity performance}
Fig.~\ref{figure:DL_channelcapacity_nokia} shows the channel capacity performance of the proposed SU-MIMO scheme against baseline methods using Nokia's mid-band dataset, with the UL SNR held constant at 10 dB. In both $N_p = \{1, 2\}$ cases, the proposed scheme exhibits a marginal performance degradation of approximately 1 dB from ``Full CSI + SVD BF,'' which is the performance upper bound achieved with optimal precoding based on the true channel. For $N_p=1$, the RLS channel estimation-based baselines show less steep slopes due to the rank deficiency of the pilot matrix.
Unlike the LMMSE-based methods that deal with this with additional information ${\mathbf{C}_{\tilde{\mathbf{h}}}}$ and $\mathbb{E}[\tilde{\mathbf{h}}]$, our proposed method solves the problem by compressing the channel matrix into a lower dimension. Also, the proposed scheme slightly outperforms those LMMSE-based strong baselines in the $N_p = \{1, 2\}$ cases, without requiring pilot matrix feedback—a significant advantage over the baseline methods. This demonstrates the efficacy of the novel deep learning-based UL CSI feedback in TDD mode, combined with implicit CSI acquisition and SU-MIMO precoding.

The effectiveness of the proposed scheme is further evidenced in Fig.~\ref{figure:DL_channelcapacity_quadriga}, where the channel capacity performance using the QuaDRiGA mid-band dataset is analyzed. As observed in Fig.~\ref{figure:DL_channelcapacity_nokia}, the proposed method shows only a minor performance loss compared to the true channel-based upper bound. Nonetheless, the performance gap between the proposed method and the other LMMSE-based benchmarks becomes larger as the DL SNR increases.

\begin{figure}
\begin{center}
\includegraphics[width=\linewidth]{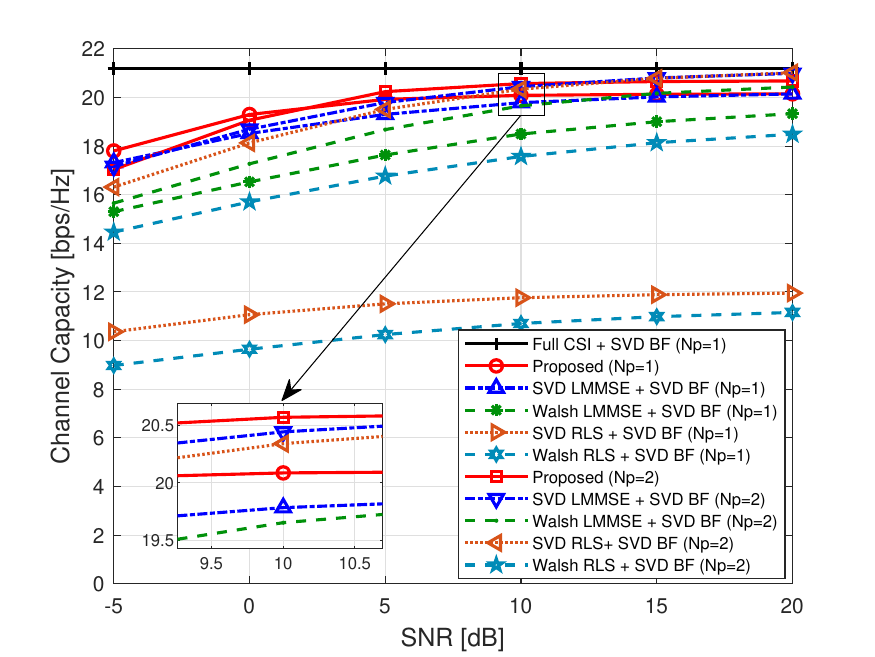}
\end{center}
\caption{DL SU-MIMO channel capacity performance versus UL SNR for $N_p=1, 2$ and $N_s=2$ with fixed DL $\text{SNR}=20\, \text{dB}$ using Nokia mid-band dataset.}
\label{figure:UL_channelcapacity_comparison_nokia}
\end{figure}

\subsubsection{Impact of UL SNR on SU-MIMO DL channel capacity} 
With the DL SNR set to 20 dB, the impact of varying UL SNR on DL channel capacity using the Nokia dataset is presented in Fig.~\ref{figure:UL_channelcapacity_comparison_nokia}. The proposed method generally surpasses the baseline approaches, and the performance gap is obvious for the $N_p=1$ case. However, beyond an UL SNR of 15 dB in the $N_p=2$ case, ``SVD LMMSE + SVD BF'' and ``SVD RLS + SVD BF'' slightly exceed the proposed scheme.

\begin{figure}
\centering
\begin{subfigure}[t]{\linewidth}
    \includegraphics[width=\linewidth]{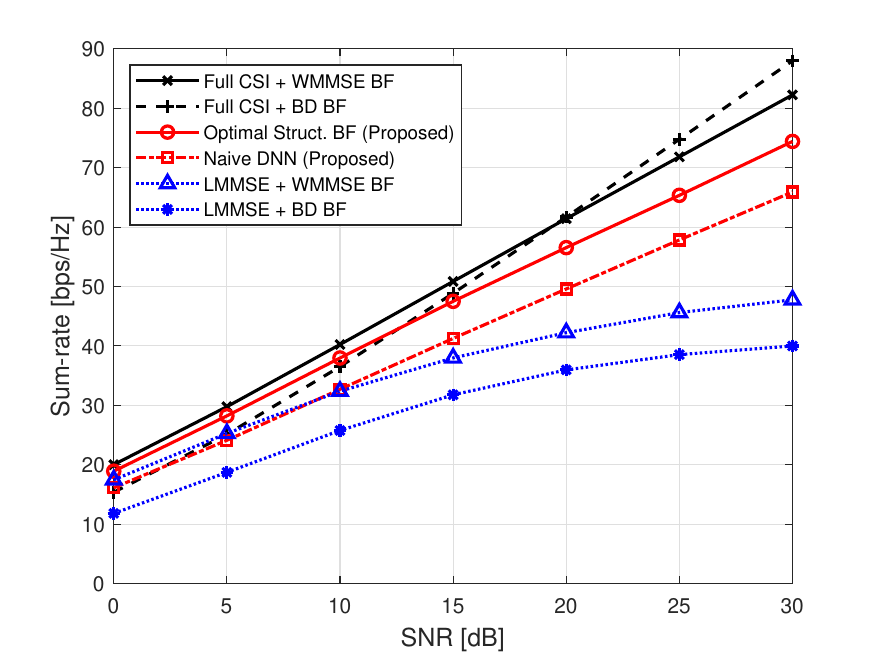}
    \caption{Using Nokia's mid-band dataset.}
    \label{figure:DL_sum-rate_nokia}
\end{subfigure}
\newline
\begin{subfigure}[t]{\linewidth}
    \includegraphics[width=\linewidth]{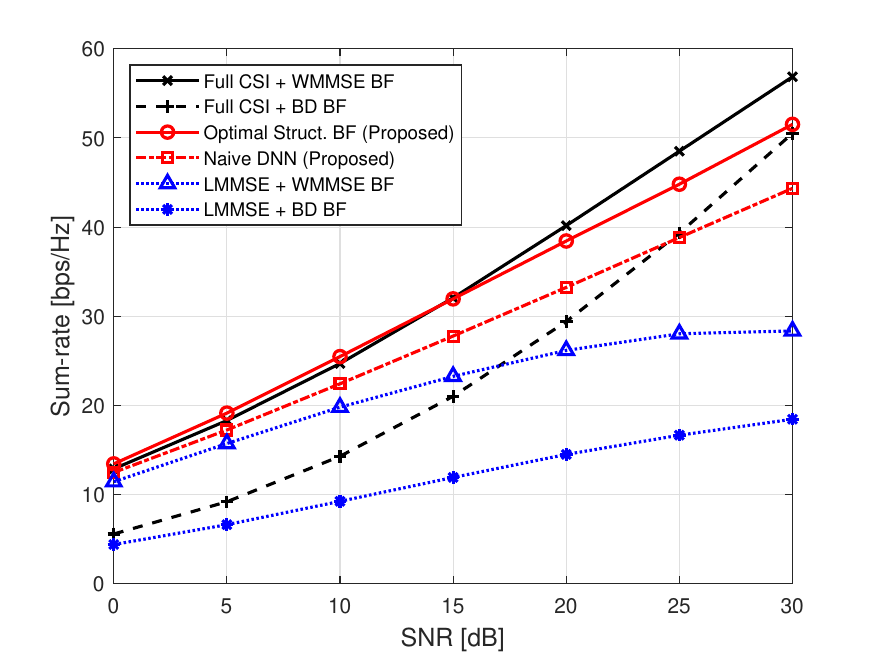}
    \caption{Using the QuaDRiGa mid-band dataset.}
    \label{figure:DL_sum-rate_quadriga}
\end{subfigure}
\caption{DL MU-MIMO sum-rate performance versus DL SNR for $ K=4 $, $ N_s=2 $, and $ N_p=1 $ with fixed UL $ \text{SNR}=10\, \text{dB} $.}
\label{fig:DL_sum-rate_comparison}
\end{figure}

\subsubsection{MU-MIMO DL sum-rate performance} 
Fig.~\ref{figure:DL_sum-rate_nokia} illustrates the DL sum-rate performance of our two proposed methods against baseline methods, using Nokia's mid-band dataset with the UL SNR fixed at 10 dB. We set the number of users $K=4$, and each UE employs $N_s=2$ streams with $N_p=1$. 
Hereafter, our analysis will focus on SVD-based pilots and compare them against LMMSE channel estimation-based methods as our primary benchmarks. 
Initially, the ``Optimal Struc. BF'' method demonstrates superior performance over the ``Full CSI + BD BF'', especially as the DL SNR increases. 
This indicates that incorporating an optimal structure-based approach into the BS DNN module provides a beneficial inductive bias and highlights the significance of such structure-based methods over naive implementations for DNN training in sum-rate maximization.
In terms of overall performance, both of our methods outperform the ``LMMSE + WMMSE BF'' and ``LMMSE + BD BF'' across all SNR ranges, except for ``Naive DNN'' at SNRs below 10 dB. Notably, ``Optimal Struc. BF'' even exceeds the ``Full CSI + BD BF'' at lower SNRs, around 10 dB. While ``Optimal Struc. BF'' does experience some performance degradation relative to Full CSI-based methods at higher SNRs, this decrease is moderate and gradual.

The DL sum-rate performance is further evaluated using the QuaDRiGa dataset in Fig.~\ref{figure:DL_sum-rate_quadriga}. While the overall trend is consistent with previous observations, ``Full CSI + BD BF'' exhibits poor performance at low SNRs. This is mainly due to the spatial consistency of channel samples in the dataset.
Given that the BD algorithm does not consider power allocation across users, it underperforms in such conditions. In contrast, this implies that the additionally introduced power allocation variables in our proposed method are effective. Note that while WMMSE BF is near-optimal, it is not the exact optimum for maximizing MU-MIMO sum-rate. This allows the ``Optimal Struc. BF'' to slightly exceed ``Full CSI + WMMSE BF'' at DL SNRs below 15 dB.

\begin{figure}
\begin{center}
\includegraphics[width=\linewidth]{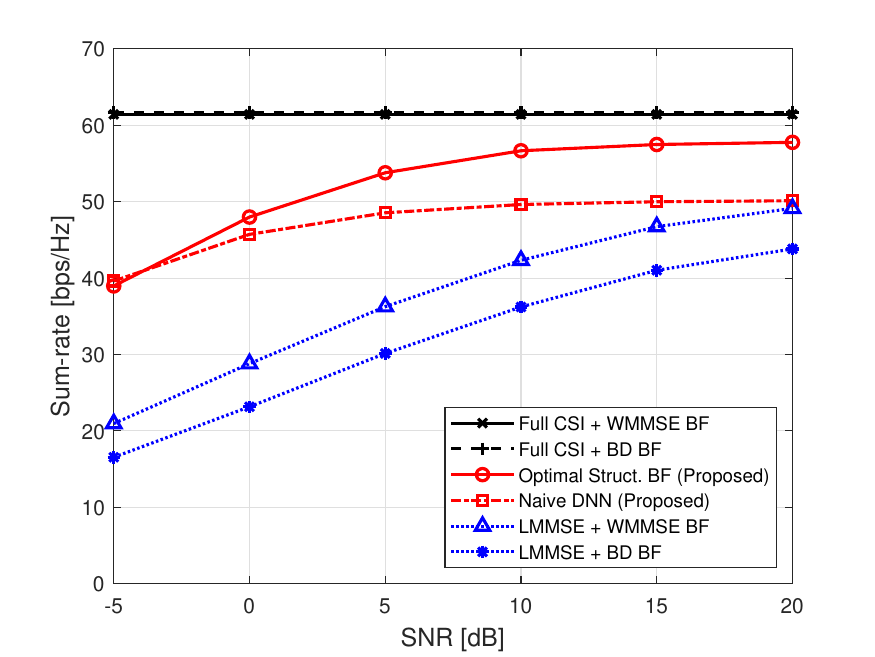}
\end{center}
\caption{DL MU-MIMO sum-rate performance versus UL SNR for $K=4$, $N_s=2$, and $N_p=1$ with fixed DL $\text{SNR}=20\, \text{dB}$ using Nokia's mid-band dataset.}
\label{figure:UL_sum-rate_comparison_nokia}
\end{figure}

\subsubsection{Impact of UL SNR on MU-MIMO DL sum-rate} 
Fig.~\ref{figure:UL_sum-rate_comparison_nokia} depicts the DL sum-rate performance of the proposed methods versus UL SNR with DL SNR set to 10 dB. Compared to the SU-MIMO case shown in Fig.~\ref{figure:UL_channelcapacity_comparison_nokia}, there is a significant performance gap between our proposed methods and the baseline approaches, and this is particularly notable at lower SNRs. 
Our proposed method shows graceful degradation with decreasing UL SNR. This demonstrates the robustness of our E2E TDD-based CSI acquisition and MU-MIMO precoding to UL noise.

\begin{figure}
    \centering
    \begin{minipage}[b]{\linewidth}
        \centering
        \includegraphics[width=0.78\linewidth]{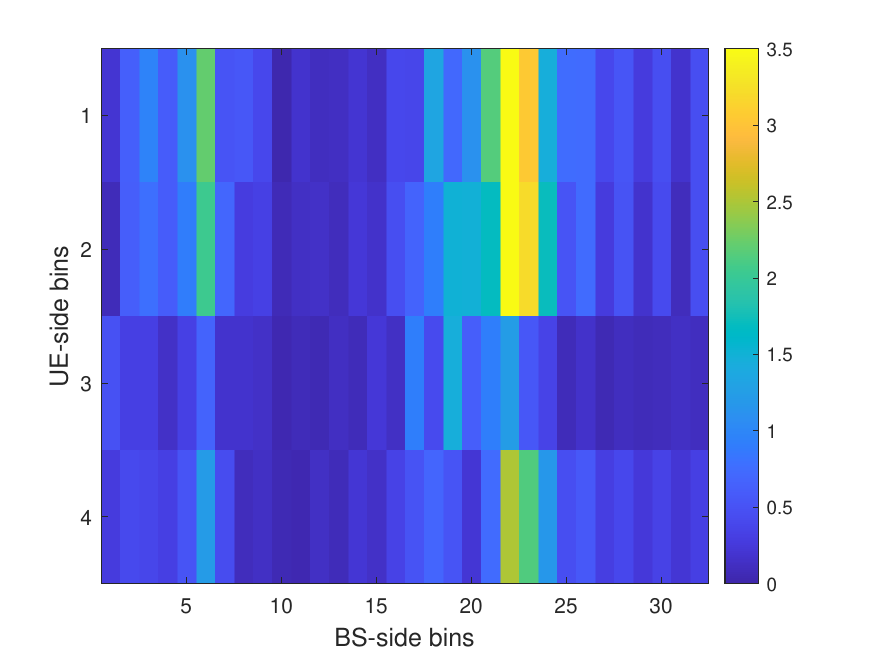}
        \\ \footnotesize(a) Channel matrix sample $\mathbf{H}$.
        \label{subfig:H_ang}
        \vspace{0.2cm}
    \end{minipage}
    \begin{minipage}[b]{\linewidth}
        \centering
        \hspace*{-0.85cm} 
        \includegraphics[width=0.70\linewidth]{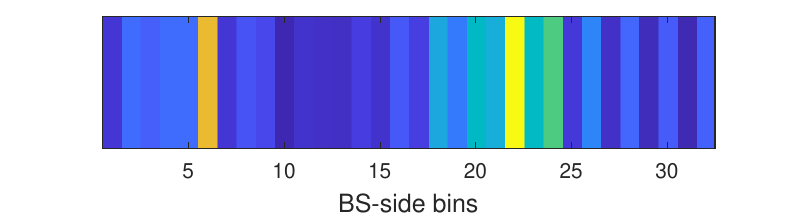}
        \\ \footnotesize(b) The compressed channel $\mathbf{H}\mathbf{p}$.
        \label{subfig:Hp}
    \end{minipage}
    \begin{minipage}[b]{\linewidth}
        \centering
        \hspace*{-0.85cm} 
        \includegraphics[width=0.70\linewidth]{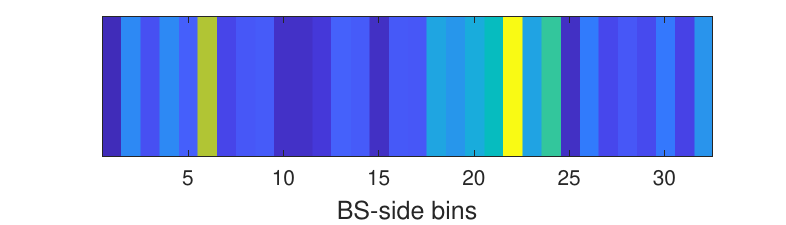}
        \\ \footnotesize(c) UL signal $\mathbf{y}^\mathrm{ul} = \mathbf{H}\mathbf{p} + \mathbf{n}^\mathrm{ul}$ (SNR = 10 dB).
        \label{subfig:y}
    \end{minipage}
    \caption{Comparison of the magnitude of the angular domain channels: the original channel matrix and its compressed versions.}
    \label{figure:compression_comparison}
\end{figure}

\begin{figure}
\begin{center}
\includegraphics[width=\linewidth]{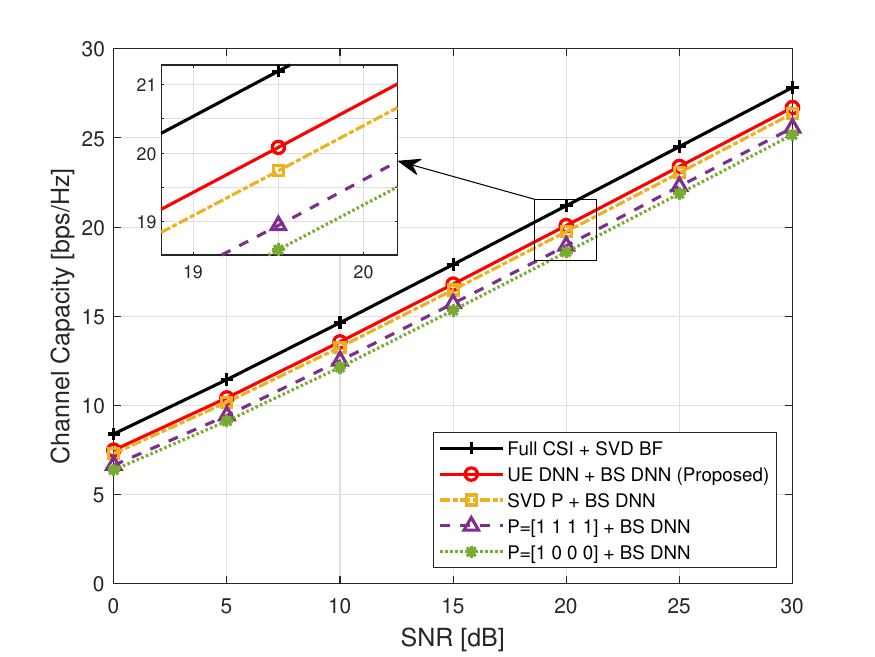}
\end{center}
\caption{DL SU-MIMO channel capacity performance versus DL SNR for $N_s = 2$ and $N_p = 1$ across different $\mathbf{P}$ matrices utilizing Nokia’s mid-band dataset.}
\label{figure:P_comparison_nokia}
\end{figure}

\begin{figure}
    \centering
    \begin{minipage}[t]{\linewidth}
        \centering
        \includegraphics[width=0.7\linewidth]{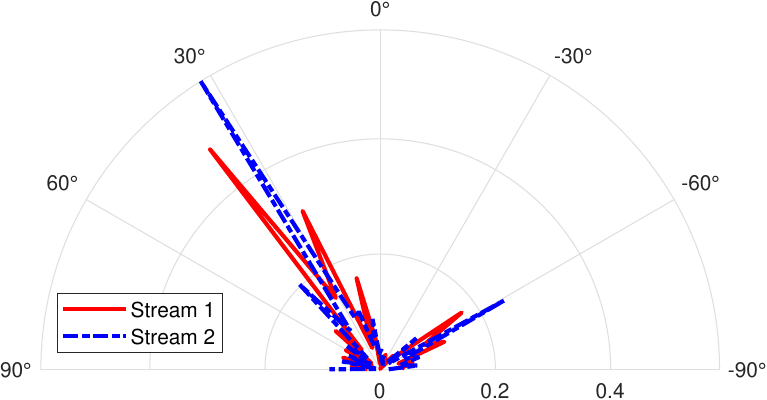}
        \\ \footnotesize(a) Beam pattern of the proposed method.
    \end{minipage}
    \newline
    \vspace{0.5cm} 
    \hfill
    \begin{minipage}[t]{\linewidth}
        \centering
        \includegraphics[width=0.7\linewidth]{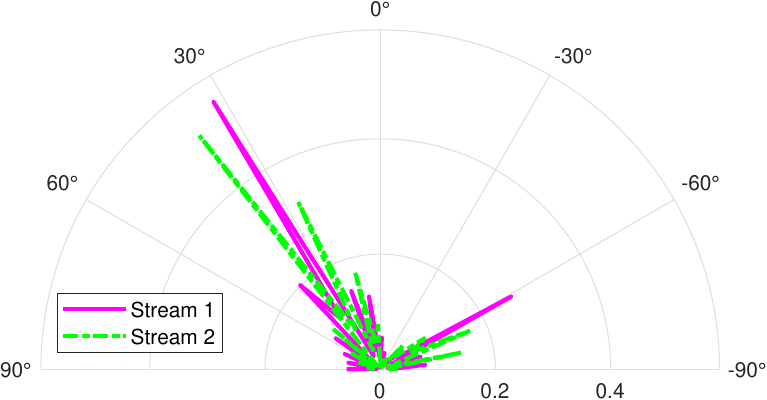}
        \\ \footnotesize(b) Optimal beam pattern.
    \end{minipage}

    \caption{Comparison of beam patterns for both the proposed and optimal methods, given an UL SNR of 10 dB and a DL SNR of 20 dB.}
    \label{figure:beampattern_comparison}
\end{figure}

\subsection{Learned channel-adaptive pilot and joint estimation and precoding mechanism}
Fig.~\ref{figure:compression_comparison} illustrates the noise-robust linear compression of the channel matrix $\mathbf{H}$ by the learned channel-adaptive pilot $\mathbf{p}$ which indicates the pilot with $N_p=1$. Fig.~\ref{figure:compression_comparison}(a) visualizes the magnitude of an angular domain channel sample from the Nokia dataset, also referred to as the virtual channel representation \cite{Sayeed:02}. The angular domain is constructed by applying array response vectors to $\mathbf{H}$, with the angles corresponding to these vectors being equally spaced in terms of their arcsine values from $-\frac{\pi}{2}$ to $\frac{\pi}{2}$. The number of bins on both the BS and UE sides is equivalent to the number of antennas. With a single pilot ($N_p=1$), the learned pilot $\mathbf{p}$ projects $\mathbf{H}$ onto a one-dimensional space within the angular domain, as shown in Fig.~\ref{figure:compression_comparison}(b). Since the channel matrix elements are complex-valued, this one-dimensional projection within the angular domain encompasses more than a mere summation of elements in each column, resulting in what can be termed a \emph{constructive angular domain projection}. Interestingly, this angular domain projection is trained without any exposure to angular domain information.
Moreover, due to the unitary invariance of noise variance, this type of projection improves noise robustness, as clearly demonstrated in Fig.~\ref{figure:compression_comparison}(c).

Furthermore, Fig.~\ref{figure:P_comparison_nokia} demonstrates the superiority of the proposed channel-adaptive pilot over other potential pilots in terms of DL channel capacity. Except for the upper bound, all other compared methods employ the BS DNN to assess the effectiveness of the pilot $\mathbf{p}$.
Since our channel-adaptive pilot performs angular domain projection, our proposed method outperforms ``$P=[1~ 1~ 1~ 1]$ + BS DNN.'' It also exceeds the performance of ``SVD P + BS DNN,'' where the SVD-based pilot matrix is a good choice for low-rank approximation.

Upon receiving this compressed UL pilot signal, the BS DNN directly computes the precoder. An example of this learned joint CSI acquisition and precoding for SU-MIMO is depicted in Fig.~\ref{figure:beampattern_comparison}, using the same channel sample $\mathbf{H}$ as in Fig.~\ref{figure:compression_comparison}. Fig.~\ref{figure:beampattern_comparison}(a) illustrates the learned beam pattern from the proposed method, and Fig.~\ref{figure:beampattern_comparison}(b) shows the optimal beam pattern derived by SVD-based precoding with WF power allocation. The two beam patterns are similar, with beam directions toward the strong angle components of the channels in Fig.~\ref{figure:compression_comparison}; however, the stream order differs as pointed out in Remark 1. While SVD is typically performed in descending order of singular values, the order does not impact the optimality of the precoder as long as the right singular vectors corresponding to the largest $N_s$ singular values are selected.



\begin{figure}
\centering
\begin{subfigure}[b]{\linewidth}
    \includegraphics[width=\linewidth]{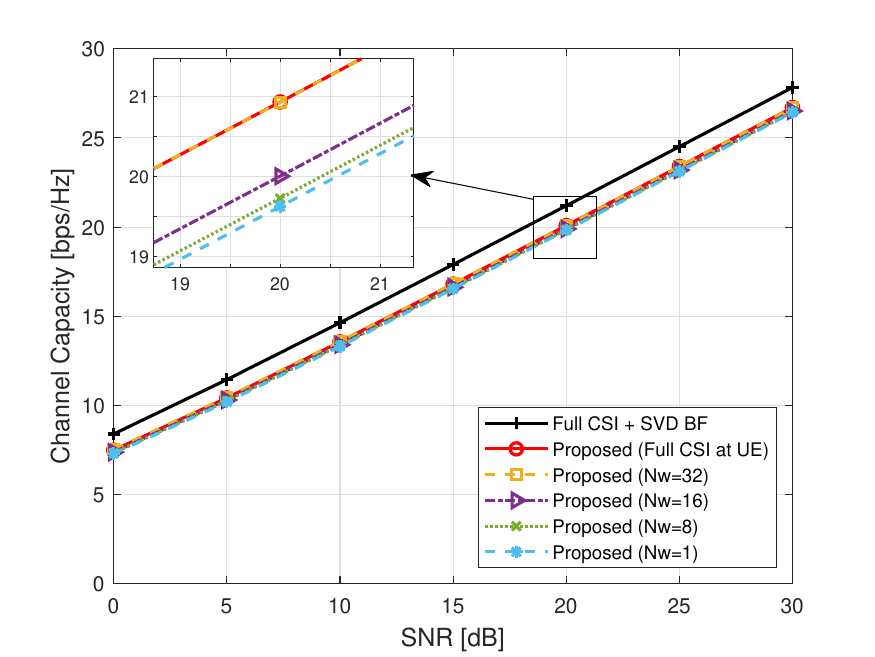}
    \caption{SU-MIMO}
    \label{figure:DL_capacity_Nw_nokia}
\end{subfigure}
\newline
\begin{subfigure}[b]{\linewidth}
    \includegraphics[width=\linewidth]{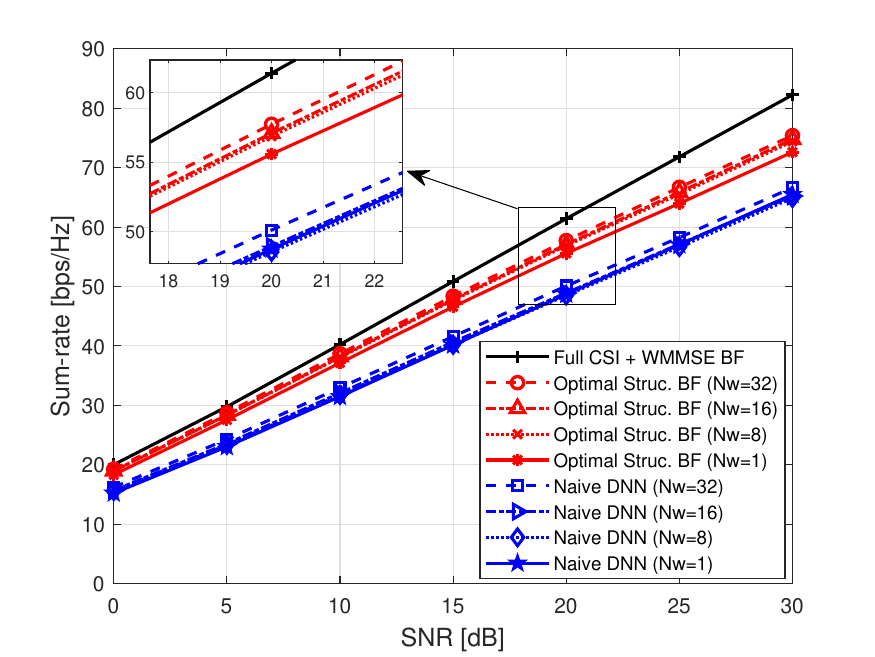}
    \caption{MU-MIMO}
    \label{figure:DL_sum_rate_Nw_nokia}
\end{subfigure}
\caption{DL SU-MIMO channel capacity and MU-MIMO sum-rate performance comparison over different numbers of probing beams, versus UL SNR for $N_s=2$ and $N_p=1$ with a fixed UL $\text{SNR}=10\, \text{dB}$, using Nokia's mid-band dataset.}

\label{fig:performance_comparison_Nw_nokia}
\end{figure}

\begin{figure}
\centering
\begin{subfigure}[b]{\linewidth}
    \centering
    \includegraphics[width=0.78\linewidth]{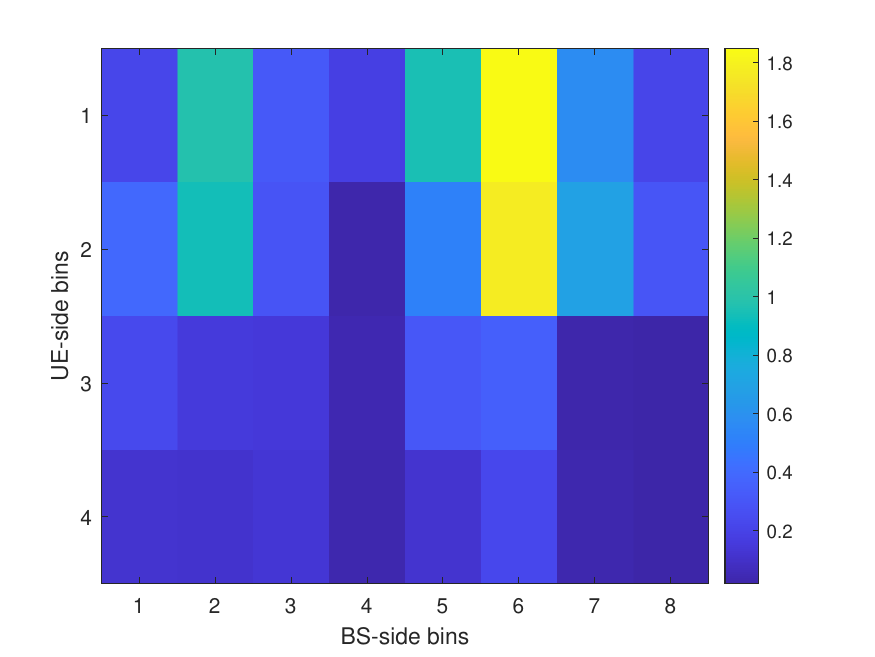}
    \caption{$N_w=8$}
    \label{subfig:DL_capacity_comparison_Nw_nokia}
\end{subfigure}
\newline
\begin{subfigure}[b]{\linewidth}
    \centering
    \includegraphics[width=0.78\linewidth]{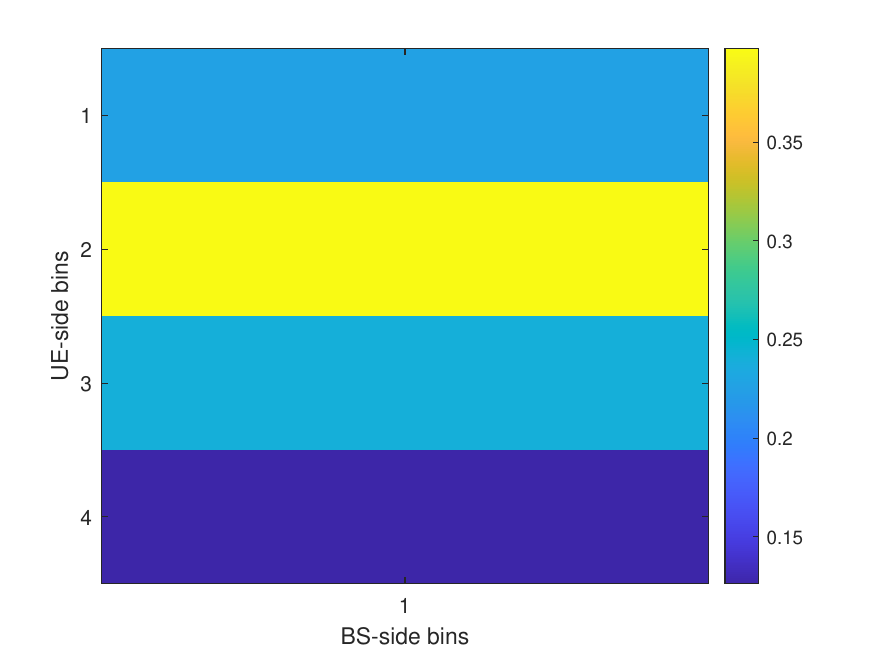}
    \caption{$N_w=1$}
    \label{subfig:DL_sum_rate_comparison_Nw_nokia}
\end{subfigure}

\caption{Comparison of noiseless angular domain representations for probing beam results with $N_w=\{1,8\}$. This illustrates the information loss for the linear compression into a lower dimension from the UE-side, as determined by the number of BS-side angular bins, which corresponds to the number of probing beams.}

\label{figure:probingbeam_CSI}
\end{figure}

\subsection{Imperfect channel knowledge at UE}
In this subsection, we relax the assumption of full channel knowledge at the UE to consider DL probing beam results instead. The received DL probing signals at UE $k$, denoted by ${\mathbf{Y}_k^{\text{prob}} \in \mathbb{C}^{N_r \times N_w}}$ is defined as
\begin{align}
{{\mathbf{Y}}^{\text{prob}}_k} = {\mathbf{H}}_k {\mathbf{A}}^{\text{prob}} + {\mathbf{N}}_k^{\text{prob}},
\label{eqn:dl_probmodel}
\end{align}
where ${\mathbf{A}}^{\text{prob}}$ indicates the discrete Fourier transform matrix with $N_w$ equally spaced probing beams. 
Instead of relying on full DL CSI at the UE as in \eqref{eqn:pilot_generation}, the UE DNN now leverages $\mathbf{Y}_k^{\text{prob}}$ as input, i.e.,
\begin{align}
{\mathbf{P}}_k = f_{\text{UE}} \left({{\mathbf{Y}}^{\text{prob}}_k} \right).
\label{eqn:pilot_generation_with_imperfect_csi}
\end{align}
Except for the imperfect CSI at each UE as \eqref{eqn:pilot_generation_with_imperfect_csi}, all other aspects of the proposed methods remain unchanged.

Fig.~\ref{figure:probingbeam_CSI}(a) and Fig.~\ref{figure:probingbeam_CSI}(b) illustrate the SU-MIMO channel capacity and MU-MIMO sum-rate performance against the DL SNR, respectively, for different numbers of probing beams $N_w$, all set with a probing SNR of 10 dB. When $N_w=32$, the performance with relaxed channel knowledge at the UE closely approximates that of the scenario with full CSI at the UE in both SU-MIMO and MU-MIMO cases. However, as $N_w$ is reduced, there is a gradual performance decline. Interestingly, only $N_w=1$ shows minor performance degradation, even though, in the MU-MIMO case, the loss is approximately 1 dB. 

This can be further understood by examining the information retained by the probing beam results. Fig.~\ref{figure:probingbeam_CSI} displays the angular domain representations of probing beam results with $N_w = 8$ and $N_w = 1$. Fig.~\ref{figure:probingbeam_CSI}(a) shows that with $N_w = 8$, the original channel information, as seen in Fig.~\ref{figure:compression_comparison}(a), is largely maintained, resulting in no significant capacity performance degradation for either SU-MIMO or MU-MIMO. In contrast, Fig.~\ref{figure:probingbeam_CSI}(b) presents a notable observation. Despite the channel information being significantly condensed with $N_w = 1$, the critical peak information is still conveyed in the UE-side bins, which may be beneficial for compressing the original channel matrix into one dimension. Nonetheless, some information about the peak in the angular domain might be lost. For example, when comparing UE-side bins 1 and 3, bin 1 contains larger channel elements than bin 3. Yet, this relationship is inverted when $N_w=1$, as depicted in Fig.~\ref{figure:probingbeam_CSI}(b).

\section{Conclusion} \label{sec:conclusion}
In this paper, we have developed E2E deep learning methods for both SU-MIMO and MU-MIMO systems, addressing the challenges posed by pilot overhead in TDD-based MIMO systems. We have successfully demonstrated a substantial reduction in UL pilot overhead while maintaining system performance. Especially, we have introduced the channel-adaptive pilot concept that linearly compresses the UL channel, while retaining only the essential information necessary for effective precoding. Furthermore, we have developed the joint CSI acquisition and precoding modules leveraging DNNs dedicated to TDD systems. Combined with the UE DNN for channel-adaptive pilot transmission, the capability of the BS DNN to process noisy, compressed UL pilots and produce precise precoders effectively highlights the potential of E2E deep learning in enhancing the performance of MIMO systems. Application to realistic upper mid-band datasets validates the practicality of these methods and also demonstrates their superiority over traditional approaches, even LMMSE channel estimation, which is based on knowledge of the statistics of the channel matrix. 

This study's findings open new avenues for future research, particularly in exploring the broader applications of channel-adaptive pilots. Given their compatibility with TDD mode, these pilots hold promise for a variety of scenarios, such as cell-free MIMO \cite{Ngo:17}. Additionally, extending our method to include receive combining and transmit precoding with hybrid architectures, prevalent in current massive MIMO systems, could enhance the practicality of our proposed approach.

\begin{appendices}
\section{Proof of Proposition~\ref{prop:1}}
\label{app:proof_prop1}
Assume $N_p \geq N_s$ and $N_p < N_r$. Given an UL channel $\mathbf{H} \in \mathbb{C}^{N_t \times N_r}$ at UE, using the singular value decomposition (SVD) yields $\mathbf{H}=\mathbf{U} \boldsymbol{\Sigma} \mathbf{V}^H$, where the singular values, denoted as $\mu_1, \mu_2, \ldots, \mu_{N_r}$, are sorted in descending order.

Choosing a pilot matrix $\mathbf{P} \in \mathbb{C}^{N_r \times N_p}\!$ as 
$\mathbf{P} = \alpha \left[\mathbf{v}_1 ,\ldots ,\mathbf{v}_{N_p}\right],$
where $\alpha$ is the power normalization parameter to satisfy $\text{Tr}({\mathbf{P} \mathbf{P}^H}) \leq \mathcal{E}_p$, the noiseless received UL analog feedback signal becomes
\begin{align}
\mathbf{H P} 
& =\mathbf{U} \boldsymbol{\Sigma} \mathbf{V}^H \mathbf{P} \\
& =\left[\mathbf{u}_1 ,\ldots, \mathbf{u}_{N_p}\right]\left[\begin{array}{lll}
\alpha \mu_1 & & \\
& \ddots & \\
& & \alpha \mu_{N_p}
\end{array}\right] \\
&= \left[\alpha \mu_1 \mathbf{u}_1 ,\ldots, \alpha \mu_{N_p} \mathbf{u}_{N_p}\right].
\end{align}

Notice that for the DL channel $\mathbf{H}^H=\mathbf{V} \boldsymbol{\Sigma} \mathbf{U}^H$, its optimal precoder is constructed as
\begin{equation}
\operatorname{diag}\left(\sqrt{{q}_1}, \ldots, \sqrt{{q}_{N_s}}\right)\left[\mathbf{u}_1, \ldots, \mathbf{u}_{N_s}\right],
\end{equation}
where the power allocation coefficients $\left\{{q}_1, \ldots, {q}_{N_s}\right\}$ are obtained using the water-filling algorithm over the singular values $\left\{\mu_1, \mu_2, \ldots, \mu_{N_s}\right\}$. Therefore, with perfect reception of $\mathbf{H P}$ at the BS, we select the first $N_s$ columns of $\mathbf{H P}$ and divide each column by $\alpha$. Thus, the optimal precoder can be derived using the water-filling algorithm over the magnitudes of the $\mathbf{H P}$ columns.

\end{appendices}

\section*{Acknowledgment}
{The authors would like to thank Jinfeng Du and Harish Viswanathan for their helpful suggestions, discussions, and feedback, and Jie Chen for providing high-quality upper midband channel datasets.}

\bibliographystyle{ieeetr}
\begingroup
\bibliography{AZREF}
\endgroup

\end{document}